\documentclass{article}
\usepackage[margin=1in]{geometry}

\usepackage{graphicx}
\usepackage{color,framed}
\usepackage{amsmath}
\usepackage{amssymb}
\usepackage{amsfonts}
\usepackage{amsopn,amsthm}
\usepackage{bbm}
\usepackage{hyperref}

\newtheorem{theorem}{\bf Theorem}%[section]
\newtheorem{lemma}{\bf Lemma}

\newtheorem{conjecture}{Conjecture}
\newtheorem{corollary}{\bf Corollary}
\newtheorem{definition}{\bf Definition}

\newtheorem{remark}{\bf Remark}

\newcommand{\supp}{\text{supp}}

\newcommand{\E}{\mathbb{E}}

\def\Pr{\text{\rm{Pr}}}

\newcommand{\fR}{\mathfrak{R}}
\newcommand{\fS}{\mathfrak{S}}

\newcommand{\Var}{\text{\rm{Var}}}

\newcommand{\wOmega}{\widetilde{\Omega}}
\newcommand{\wPi}{\widetilde{\Pi}}

\newcommand{\PR}{\text{\rm{PR}}}
\newcommand{\CHSH}{\text{\rm{CHSH}}}

\newcommand{\womega}{\tilde{\omega}}
\newcommand{\wpi}{\tilde{\pi}}

\newcommand{\wa}{\tilde{a}}
\newcommand{\wb}{\tilde{b}}
\newcommand{\wx}{\tilde{x}}
\newcommand{\wy}{\tilde{y}}
\newcommand{\wt}{\tilde{t}}

\newcommand{\wA}{\widetilde{A}}
\newcommand{\wB}{\widetilde{B}}
\newcommand{\wX}{\widetilde{X}}
\newcommand{\wY}{\widetilde{Y}}
\newcommand{\wT}{\widetilde{T}}
\newcommand{\wS}{\widetilde{S}}

\def\Pr{\text{\rm{Pr}}}

\begin{document}

\date{}
\title{Monotone Measures for Non-Local Correlations}
\author{Salman Beigi$^1$ and Amin Gohari$^{1,2}$
\\ \emph{\small $^1$School of Mathematics, Institute for Research in Fundamental Sciences (IPM), Tehran, Iran}\\
\emph{\small $^2$Department of Electrical Engineering, Sharif University of Technology, Tehran, Iran}}

\date{January 11, 2017}

\maketitle

\begin{abstract}
Non-locality is the phenomenon of observing strong correlations among the  outcomes of local measurements of a multipartite physical system. No-signaling boxes are the abstract objects for studying non-locality, and wirings are local operations on the space of no-signaling boxes. This means that, no matter how non-local the nature is, the set of physical non-local correlations must be closed under wirings. Then, one approach to identify the non-locality of nature is to characterize closed sets of non-local correlations.
Although non-trivial examples of wirings of no-signaling boxes are known, there is no systematic way to study wirings. In particular, given a set of no-signaling boxes, we do not know a general method to prove that it is closed under wirings. 
In this paper, we propose the first general method to construct such closed sets of non-local correlations. We show that a well-known measure of correlation, called \emph{maximal correlation}, when appropriately defined for non-local correlations, is monotonically decreasing under wirings. 
 This establishes a conjecture about the impossibility of simulating isotropic boxes from each other, implying the existence of a continuum of closed sets of non-local boxes under wirings. To prove our main result, we introduce some mathematical tools that may be of independent interest:  we define a notion of \emph{maximal correlation ribbon} as a generalization of maximal correlation, and provide a connection between it and a known object called \emph{hypercontractivity ribbon}; we show that these two ribbons are monotone under wirings too.
 \end{abstract}

\section{Introduction} 
Non-locality is one of the intriguing features of nature. As predicted by quantum theory and confirmed by experiments, outcomes of measurements on subsystems of a bipartite quantum system can be correlated in a non-local way. However, there are restrictions to this non-locality, which raises the question of  which non-local  correlations are feasible in nature. The Hilbert space formalism of quantum mechanics gives some answers to this question. Nevertheless, non-locality is a more fundamental feature of nature compared to the mathematical postulates of quantum physics. So the question is whether we can characterize the limit of non-locality of nature based on more fundamental principles.

This question was first raised by Popescu and Rohrlich in~\cite{PR94} where \emph{no-signaling}, i.e., the impossibility of instantaneous communication, is proposed as a fundamental physical principle to limit non-locality. They showed that no-signaling is not strong enough to characterize non-local correlations of quantum physics.
Moreover, there are strong evidences against the possibility of realization of such highly non-local correlations in nature~\cite{vanDam}. Subsequently, other principles were proposed to characterize non-locality, see e.g.,~\cite{IC09, FSABCLA, BBLMTU, LPSW, NPA07, ABLP09, MW2009, New-A2, New-A9}.  In this paper, we provide a systematic method for studying ``closed sets of correlations," introduced as a fundamental concept in~\cite{ABLP09} to characterize non-local correlations in physical theories.

Non-local correlations are generated by locally measuring subsystems of a bipartite system (See Fig.~\ref{fig1}). Imagine that subsystems of a bipartite physical system are held by two parties, say Alice and Bob. They can decide to apply a measurement on their subsystems; these choices of measurement settings by Alice and Bob are denoted by $x$ and $y$ respectively. Letting the measurement outcomes be $a$ and $b$, in its full general case, the probability of these outcomes come from some conditional distribution
$p(ab|xy)$. We may think of this setting as a \emph{box} with two parts. Each part has an input and an output. Alice who holds the first part can choose its input $x$, and receive its output $a$. Similarly Bob who holds the second part, can choose its input $y$, and receive its output $b$.
With this notation, the no-signaling principle states that $p(a|xy)=p(a|x)$ and $p(b|xy)=p(b|y)$. Equation $p(a|xy)=p(a|x)$ for instance implies that when Alice's input $x$ to the box is specified, the distribution of the outcome $a$ does not depend on Bob's choice of input $y$.

If subsystems of the bipartite physical system are completely independent, their measurement outcomes are independent of each other. In this case, we must have $p(ab|xy)=p(a|x)p(b|y)$. These correlations as well as their convex combinations, which correspond to classically correlated subsystems via hidden variables, are called~\emph{local}. Bell's theorem and its experimental tests  suggest that correlations that are not local (non-local correlations) also exist in nature.

Important examples of no-signaling boxes include \emph{isotropic boxes}.
It is a bipartite box with binary inputs and outputs (i.e., $a, b, x, y\in \{0,1\}$) defined by:
\begin{align}\label{eq:iso-box}
\PR_{\eta}(a,b|x,y) := \begin{cases}
\frac{1+\eta}{4} \qquad\qquad \text{if } a \oplus b= xy,\\
\frac{1-\eta}{4} \qquad\qquad \text{otherwise}.\\
\end{cases}
\end{align}
The box $\PR_{\eta}$ with $0\leq \eta\leq 1/2$ is local, and with $0\leq \eta\leq 1/\sqrt{2}$ is realizable within quantum mechanics. Nonetheless, $\PR_{\eta}$ for any $0\leq \eta\leq 1$ is no-signaling. Thus a natural question is: what is the largest possible $\eta$ such that $\PR_{\eta}$ is feasible in nature~\cite{vanDam,  IC09, FSABCLA, BBLMTU, LPSW, CSS2010, OCB2012}.

\begin{figure}
\begin{center}
\includegraphics[width=2.3in]{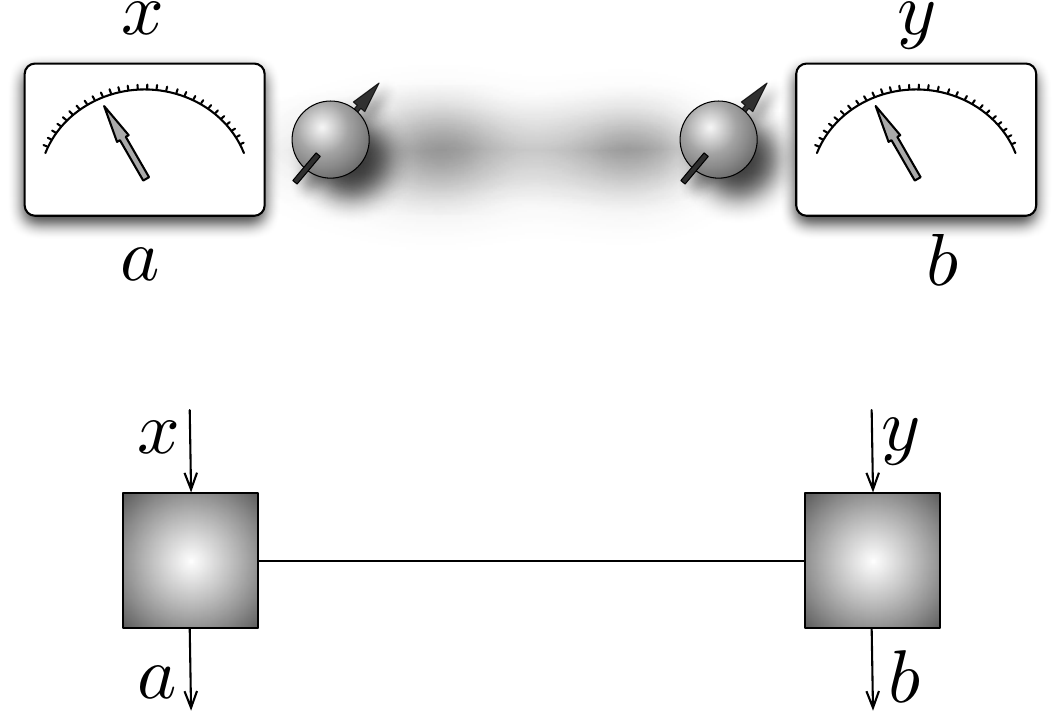}
\caption{\small Imagine that two parties share subsystems of a bipartite physical system that can be correlated. Each party can apply a measurement on her subsystem by tuning her measurement device based on some parameter, and obtain the measurement outcome. We represent the measurement parameters by $x, y$, and the measurement outcomes by $a, b$. Then in its full general case, the  outcomes $a, b$, under measurements $x, y$ are obtained with some conditional probability $p(ab|xy)$. We can think of this setting as a box with two parts, where each part has an input and an output. Given inputs $x,y$ the outputs of the box are $a, b$ with probability $p(ab|xy)$.
}
\label{fig1}
\end{center}
\end{figure}

Allcock et al.\ propose the concept of ``closed sets of correlations" to study the set of realizable non-local boxes \cite{ABLP09}. They observe that no matter how non-local nature is, the set of non-local boxes must be closed under certain local operations, called \emph{wirings}~\cite{Wiring1, Wiring2}.
 To illustrate the idea of wirings, we describe it in a simple case, involving only two boxes.

Having two boxes, each party can choose the input of the second box as a function of the output of the first box. More precisely, denoting the inputs and outputs of the two boxes by subscripts $1, 2$, Alice can first choose $x_1$ arbitrarily and use the first box to generate an output $a_1$. Then she may put $x_2=a_1$, i.e., she may {wire} the output of the first box to the input of the second box. Bob can similarly use the output of the first box to determine the input of the second box.
With these wirings, the parties generate a new box $p(a_2b_2|x_1y_1)$. That is, combining two boxes with wirings they generate a new box under local operations.

Due to the operational definition of wirings, no matter how non-local nature is, the space of physical boxes must be closed under wirings~\cite{ABLP09}. Then to characterize non-locality in nature, we may first look for subsets of no-signaling boxes that are closed under wirings. Three examples of such closed sets are the set of local correlations, the set of quantum correlations and the whole set of no-signaling boxes. 
Other examples of closed sets of no-signaling   boxes are described in~\cite{MW2009, LOP2} and in~\cite{Review14}. However, as noted in~\cite{Review14}, characterizing sets of boxes that are closed under wirings is a difficult problem in general.   A source of difficulty is that there is no limit on the number of boxes that the two parties may choose to use, and the number of possible ways to wire these boxes (defined more precisely later) grows exponentially in the number of boxes. Therefore, having even a few boxes as our resource, it is a difficult problem to discern whether a target box can be simulated via an appropriate wiring or not. It is not even known whether this problem is decidable or not~\cite{Review14}.

In~\cite{Review14} it is asked whether there exists a continuum of sets of non-local boxes that are closed under wirings. In particular, it is conjectured that:

\begin{conjecture}[\cite{Review14}] \label{conj1}
For $1/2< \eta_1<\eta_2<1$, two parties cannot use common randomness and an arbitrary number of copies of $\PR_{\eta_1}$ to generate a single copy of $\PR_{\eta_2}$ with wirings.
\end{conjecture}

Although some partial results on the above conjecture have been found~\cite{Short09, Forster11, DukaricWolf}, no general method for studying wirings is known. Even with the simple structure of isotropic boxes, we do not know how to characterize the \emph{closure} of an isotropic box $\PR_\eta$ under wirings, i.e., the set of boxes that can be obtained by wirings of some copies of $\PR_\eta$.

In this paper we present a systematic method for constructing sets of non-local boxes that are closed under wirings. We introduce an invariant of non-local boxes that is monotone under wirings. Our parameter is in terms of a well-known measure of correlation called \emph{maximal correlation}~\cite{Hirschfeld, Gebelein, Renyi1, Renyi2, Witsenhausen}. 
We show that maximal correlation, when appropriately defined for non-local boxes, cannot increase under wirings, i.e., maximal correlation is a monotone under wirings. With this result, we can explicitly construct sets of non-local boxes that are closed under wirings.
Moreover, by computing maximal correlation for isotropic boxes, we prove Conjecture~\ref{conj1} for the range of parameters $1/\sqrt 2\leq \eta_1<\eta_2<1$.
Our result, in particular implies that there is a continuum of sets of no-signaling boxes that are closed under wirings.

In the rest of this section we briefly discuss the main definitions and ideas of the paper, and informally state the main results. A reader interested only in the statements of the main results, but not their proofs, may continue reading this section and ignore the rest of the paper.

\subsection{Maximal correlation} 
To get an insight into the types of measures of correlation that are useful for studying wirings, consider the similar but simpler problem of simulating a joint distribution from another. More specifically, suppose that we are given two bipartite probability distributions $p_{AB}$ and $q_{A'B'}$. The question is, given an arbitrary number copies of (samples from) $p_{AB}$ can we generate a single copy of (a sample from) $q_{A'B'}$ by only employing local operations on the $A$ parts and $B$ parts separately? This is a hard problem in general since we assume that the number of available copies of the resource distribution $p_{AB}$ is arbitrarily large. 

One may attack the above problem by showing that $q_{A'B'}$ is \emph{more correlated} that $p_{AB}$, so $p_{AB}$ cannot be transformed to $q_{A'B'}$ under local operations. This strategy depends on the \emph{measure} of correlation that we use. The point is that we are allowed to use an arbitrary number copies of $p_{AB}$. Moreover, for most measures of correlations (including mutual information), if $p_{AB}$ has some positive correlation, the correlation of $p^n_{AB}$, i.e., $n$ i.i.d.\ copies of $p_{AB}$, goes to infinity as $n$ gets larger and larger. Then this strategy fails for usual measures of correlation. However,  there is a  measure of correlation, called maximal correlation that can be used for this problem.\\

\noindent\textbf{Maximal correlation.}
Given  a bipartite probability distribution $p_{AB}$, its maximal correlation denoted by $\rho(A, B)$ is the maximum of Pearson's correlation coefficient over all functions of $A$ and $B$, i.e.,
\begin{align}\label{eq:max-correlatoin-31}
\rho(A, B):=& \max \frac{\E[(f_A-\E[f_A])(g_B-\E[g_B])]}{\Var[f_A]^{1/2}\Var[g_B]^{1/2}}
\end{align}
where $\E[\cdot]$ and $\Var[\cdot]$ are expectation value and variance respectively. Moreover, the maximum is taken over all non-constant functions $f_A, g_B$ of $A$ and $B$ respectively.

We always have $0\leq \rho(A, B)\leq 1$. Moreover, $\rho(A, B)=0$ if and only if $A$ and $B$ are independent, and $\rho(A, B)=1$ if and only if $A$ and $B$ have a \emph{common data}~\cite{Witsenhausen}.
Maximal correlation can be computed efficiently by diagonalizing a certain matrix~\cite{Kumar2, KangUlukus}.

Maximal correlation has the intriguing property that
\begin{align}\rho(AA', BB') = \max\{\rho(A, B), \rho(A', B')\},\label{eqn:amnewtensor}\end{align}
when $AB$ and $A'B'$ are independent, i.e., $p_{AA'BB'}=p_{AB}\cdot p_{A'B'}$.
This property is sometimes called the \emph{tensorization property}.
Moreover, as a measure of correlation, maximal correlation is monotone under local operations. That is, if $q_{A'B'}$ can be generated from $p_{AB}$ under local stochastic maps, then $\rho(A', B')\leq \rho(A, B)$.

Given the above two properties, we conclude that if for two distributions $p_{AB}$ and $q_{A'B'}$ we have $\rho(A, B)< \rho(A',B')$, then $q_{A'B'}$ cannot be generated locally even if an arbitrary large number of copies of $p_{AB}$ is available.  \\

\noindent\textbf{Maximal correlation for non-local boxes.}
Given a no-signaling box determined by conditional distributions $p(ab|xy)$ we define its maximal correlation by
$$\rho(A, B|X, Y):=\max_{x, y} \rho(A, B|X=x, Y=y).$$
That is, any $(X, Y)=(x, y)$ induces a distribution on $A, B$, so we may compute the maximal correlation $\rho(A, B|X=x, Y=y)$ of this conditional distribution. The maximal correlation of the box is the maximum of all these numbers. Since maximal correlation of bipartite distributions can be computed efficiently, the maximal correlation of non-local boxes can be computed efficiently too.
Our main result in this paper is that maximal correlation of non-local boxes is monotone under wirings.

Suppose that the Alice and Bob share a no-signaling box $p(ab|xy)$ and have some a priori correlation. Thus they can choose their inputs of the box according to their a priori correlation, i.e., with respect to  some distribution $q_{XY}$. With thse random choices of inputs, they obtain the joint distribution $q(abxy):=q(xy)p(ab|xy)$ on inputs and outputs of the box. A simple argument shows that (see Appendix \ref{mc-wiring2})
\begin{align}\label{eq:AX-BY-maximal-correlation}
\rho(A, B)\leq \rho(AX, BY)\leq \max\{\rho(X, Y), \rho(A, B|X, Y)\},
\end{align}
where $\rho(X, Y)$ is computed with respect to the distribution $q_{XY}$, and $\rho(A, B|X, Y)$ is the maximal correlation of the box. This means that the maximal correlation between the outputs of the box, is bounded by the maximum of the a priori maximal correlation between Alice and Bob and the maximal correlation of the box shared between them.

Now, assume that Alice and Bob are provided with $n$ boxes $p_i(a_ib_i|x_iy_i)$ for $1\leq i\leq n$. Suppose that they wire these boxes by taking the outputs of the $i$-th box, and feeding them into the input of box $i+1$. That is, after using the $i$-th box, Alice and Bob obtain outputs $a_i$ and $b_i$ respectively, and then use box $i+1$ by setting its inputs $x_{i+1}=a_i$ and $y_{i+1}=b_i$. This is a very special way of wiring of boxes and will be generalized later. With this wiring, Alice and Bob generate a box $q(a_nb_n|x_1y_1)$.
Then we  claim that
\begin{align}\label{eq:rho-an-bn-x1-y1}
\rho(A_n, B_n|X_1, Y_1)\leq \max_i \rho(A_i, B_i|X_i, Y_i).
\end{align}
Let us first prove this inequality for $n=2$. Fix the inputs of the first box $(X_1, Y_1)=(x_1, y_1)$. After using this box, Alice and Bob  put $(x_2, y_2)=(a_1, b_1)$ which are picked with probability $p_1(a_1b_1|x_1y_1)$; in other words, the distribution of the inputs of the second box is $p_1(x_2y_2|x_1y_1)=p_1(a_1b_1|x_1y_1)$. Then using~\eqref{eq:AX-BY-maximal-correlation} we have
\begin{align*}\rho(&A_2, B_2|X_1=x_1, Y_1=y_1)\leq
\\& \max\big\{  \rho(A_1, B_1|X_1=x_1, Y_1=y_1), \rho(A_2, B_2|X_2, Y_2)   \big\}.\end{align*}
Since this inequality holds for all $(x_1, y_1)$, by the definition of maximal correlation for boxes, equation~\eqref{eq:rho-an-bn-x1-y1} holds. This inequality for arbitrary $n$ is proved by the same argument and a simple induction.

Equation~\eqref{eq:rho-an-bn-x1-y1} states that by wiring of no-signaling boxes in the particular way described above, one cannot generate a box with a larger value of maximal correlation comparing to those of available boxes. That is, maximal correlation of boxes is monotone under this special type of wirings of no-signaling boxes.\\

\noindent\textbf{Wirings of no-signaling boxes.}
So far we assumed that in wirings, each party sets the input of a box to be equal to the output of the previous box. However, in general each input can be chosen as a possibly random function of \emph{all} the previous inputs and outputs, i.e., for instance Alice to determine input $x_i$ of the $i$-th box can apply a stochastic map on $x_1, a_1, \dots, x_{i-1}, a_{i-1}$. There are interesting examples of wirings of this type in the literature~\cite{vanDam, BS09}.
The above argument can be modified to prove~\eqref{eq:rho-an-bn-x1-y1} even for these types of wirings. Nevertheless,
 wirings of non-local boxes can be even more complicated.

 By the no-signaling condition, the parties can use their available boxes in different orders. Each party can choose an arbitrary ordering of boxes and wire the output of a box to the input of the next box in that order. This point is justified by the no-signaling condition, and can intuitively be verified by thinking of the local use of boxes as making measurements on subsystems of a bipartite physical system. Such measurements can be done asynchronously. See Fig.~\ref{fig2} for an example of wirings of three boxes in different orders.

\begin{figure}
\begin{center}
\includegraphics[width=3in]{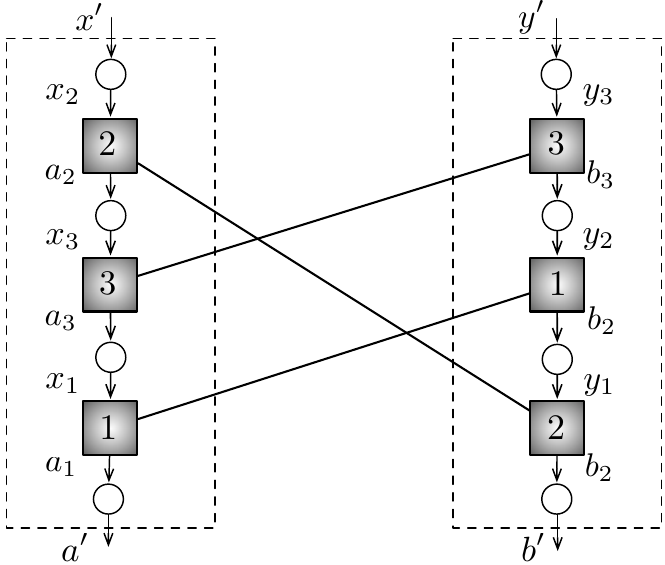}
\caption{\small The parties can wire their available boxes in a non-trivial order. The fact that boxes can be used in different orders by the parties is a consequence of the no-signaling principle.  Here the first party, Alice, uses the boxes in order $2, 3,1$. She has some input $x'$, based on which picks $x_2$, the input of box $2$. Here the circle on the top depicts the stochastic function that maps $x'$ to $x_2$. After using box $2$, Alice obtains output $a_2$. Then she has $x', x_2$ and $a_2$ in hand, and applies another stochastic maps to generate $x_3$, the input of box $3$. She continues until using all the boxes. At the end she applies the final stochastic map  to determine her final output $a'$. Bob uses the boxes in order $3,1,2$ and performs similarly.}
\label{fig2}
\end{center}
\end{figure}

A further degree of freedom in wirings is the very choice of the order of boxes used by a party, as that itself can depend on the outputs of boxes they have already used. For instance, depending on the output of the first box, a party may choose the next box to be used.
Again combining some boxes, the parties can generate a new box under local operations. A formal definition of wirings comes in Section~\ref{sec:wiringn}.

Now we may raise the question of the monotonicity of maximal correlation under wirings for these generalized types of wirings. We argued that if the two parties use boxes in the same order that is fixed, then the maximal correlation of the new box generated under wirings, is at most the maximum of the maximal correlations of the available boxes. The question is whether we can prove the same result when the boxes are used in random orders.

One of the main results of this paper is answering the above question in the affirmative. Nevertheless, our proof of this fact is not as simple as the proof in the special case of equation~\eqref{eq:rho-an-bn-x1-y1}. Our proof in this paper is an indirect one that uses two other measures of correlation with the tensorization property.

\subsection{Hypercontractivity ribbon} 

Correlation can also be measured via hypercontractivity inequalities. For 
a bipartite probability distribution $p_{AB}$, its hypercontractivity ribbon (HC~ribbon) which we denote by $\mathfrak R(p_{AB})$, is a subset of the real plane consisting of pairs $(1/\alpha, 1/\beta)\in [0,1]^2$ such that  $\E[f_A g_B]\leq \|f_A\|_{\alpha}\|g_B\|_{\beta}$ for all functions $f_A, g_B$, where $\|\cdot\|_{\alpha}, \|\cdot\|_{\beta}$ are Schatten norms. More precisely, 
\begin{align*}\fR(p_{AB}) := \big\{(1/\alpha, 1/\beta)\in [0,1]^2|\,  \forall f_A, g_B\,\, \E[f_A g_B]\leq \|f_A\|_{\alpha}\|g_B\|_{\beta}\big\}.\end{align*}

By definition $\fR(A, B)\subseteq [0,1]^2$. Moreover,
when $A$ and $B$ are independent, we have $\fS(A, B)=[0,1]^2$. That is, when $A$ and $B$ are independent, the HC~ribbon is the largest possible. Indeed, the more correlated $A$ and $B$ are, the smaller their HC~ribbon is. 
More precisely, the HC~ribbon is a measure of correlation in the following sense: if $q_{A'B'}$ can be obtained from $p_{AB}$ under local operations, then we have
\begin{align}\label{eqn:data-hc-5}
\fR(p_{AB})\subseteq \fR(q_{A'B'}).
\end{align}
Moreover, we have 
\begin{align}\mathfrak R(p^n_{AB}) = \fR(p_{AB}).\label{eqn:tensorization-HC-00}\end{align}
That is, the HC~ribbon also has the tensorization property.
Putting these together we conclude that if $\fR(p_{AB})$ is not contained in $\fR(q_{A'B'})$, the former cannot be transformed to the latter under local operations even if an arbitrary number of i.i.d.\ copies of $p_{AB}$ is available.

HC~ribbon was originally defined by Ahlswede and G\'acs in~\cite{AhlswedeGacs}, and has found applications in information theory (e.g., see \cite{KamathAnantharam, Anantharametal}). HC~ribbon in the quantum case is defined and studied in~\cite{DelgoshaBeigi}.

A remarkable equivalent characterization of the HC~ribbon was recently found by Nair~\cite{Nair}. He showed that $\fR(p_{AB})$ is indeed the set of pairs $(\lambda_1, \lambda_2)$ of non-negative numbers such that for all auxiliary random variables $U$ (i.e., all $p_{UAB}=p_{AB}\cdot p_{U|AB}$ for an arbitrary conditional distribution $p_{U|AB}$), we have
\begin{align}\label{eq:intro-nair-eq-00}
I(U; AB)\geq \lambda_1 I(U; A) + \lambda_2 I(U; B).
\end{align}
Here $I(\cdot\,; \cdot)$ is the mutual information function defined as $I(A;B)=\sum_{a,b}p(a,b)\log\frac{p(a,b)}{p(a)p(b)}$.\\

\noindent\textbf{HC~ribbon for non-local boxes.} HC~ribbon can be defined for no-signaling boxes as well. Given a box determined by conditional distributions $p(ab|xy)$ we define its HC~ribbon by
\begin{align}\label{eq:box-hc-ribbon-00}
\fR(A, B|X, Y)=\bigcap_{x,y} \fR(A, B|X=x, Y=y).
\end{align}
We show in this paper that the HC~ribbon is monotone under wirings of no-signaling boxes. That is, given $n$ boxes $p_i(a_ib_i|x_iy_i)$, if we can generate another box $q(a'b'|x'y')$ from wirings of these boxes, then we have
$$\bigcap_{i=1}^n \fR(A_i, B_i|X_i, Y_i) \subseteq \fR(A', B'|X', Y').$$
This means that wirings of no-signaling boxes can only expand the HC~ribbon. The proof of this fact is relatively involved and is one of the main technical contributions of this paper. Here we only mention that the main tool that we use in this proof is the chain rule of mutual information.

\subsection{Maximal correlation ribbon}

Let us turn back to the problem of the monotonicity of maximal correlation of boxes under wirings. 
To prove this fact, we found it easier to work with a generalization of maximal correlation, that we define for the first time and call it \emph{maximal correlation ribbon} (MC~ribbon). We remark that even though we define the MC~ribbon for our purposes here, it is of independent interest. 

The MC~ribbon is a subset of the real plane, defined as follows:
\begin{align}
\fS(p_{AB}):=\big\{&(\lambda_1, \lambda_2)\in [0,1]^2|\forall f_{AB}\,\, \Var[f]\geq \lambda_1\Var_A\E_{B|A}[f]+\lambda_2 \Var_B\E_{A|B}[f]\big\}.\label{eq:mc-rib-00}
\end{align}
Here, for instance, $\E_{B|A}[\cdot]$ denotes the conditional expectation value. Thus $\E_{B|A}[f]$ is a function of $A$, and then $\Var_A\E_{B|A}[f]$ makes sense.

The MC~ribbon has similar properties as the HC~ribbon. It is not hard to verify that $\fS(A, B)\subseteq [0,1]^2$, and that equality holds only if $A$ and $B$ are independent. Moreover, the MC~ribbon is a measure of correlation that satisfies the monotonicity and the tensorization properties in the sense of~\eqref{eqn:data-hc-5} and~\eqref{eqn:tensorization-HC-00}.\\

\noindent\textbf{MC~ribbon and HC~ribbon.}
The definition of the MC~ribbon in~\eqref{eq:mc-rib-00} is similar to Nair's characterization of the HC~ribbon given by~\eqref{eq:intro-nair-eq-00}. Indeed the MC~ribbon is defined by replacing mutual informations in~\eqref{eq:intro-nair-eq-00} by variances. 
Another contribution of this paper, which again is of independent interest, is that the HC~ribbon is always contained in the MC~ribbon:
$$\fR(A, B)\subseteq \fS(A, B).$$
The connection between the MC~ribbon and the HC~ribbon will be discussed more precisely in Subsection~\ref{subsec:MC-HC-ribbons}.\\

\noindent\textbf{MC~ribbon for non-local boxes.} The MC~ribbon for non-local boxes can be defined similarly to equation~\eqref{eq:box-hc-ribbon-00}: 
$$\fS(A, B|X, Y)=\bigcap_{x,y} \fS(A, B|X=x, Y=y).$$
We prove that, similar to the HC~ribbon, the MC~ribbon can only expand under wirings. To prove this result we use the connection between the MC~ribbon and the HC~ribbon mentioned above.\\

\noindent\textbf{MC~ribbon vs maximal correlation.}
Using all the above tools we can then prove our first claim, that maximal correlation of non-local boxes is monotone under wirings. To prove this claim we first show that maximal correlation can be characterized in terms of the MC~ribbon. More precisely, for any bipartite distribution $p_{AB}$ we have 
$$ \rho^2(A, B)=\inf\frac{1-\lambda_1}{\lambda_2},$$
where infimum is taken over all $(\lambda_1, \lambda_2)\in \fS(A, B)$ with $\lambda_2\neq 0$. With this characterization of maximal correlation, and the monotonicity of the MC~ribbon under wirings, the monotonicity of maximal correlation under wirings is immediate.

\subsection{Proof of Conjecture~\ref{conj1}}
As an application of the above results, we study Conjecture~\ref{conj1}. The maximal correlation of box $\PR_\eta$ is equal to $\eta$. Then using the fact that maximal correlation cannot be increased under wirings, Conjecture~\ref{conj1} is proved in the special case where the parties are not provided with common randomness. To study the case where common randomness is also available, we employ the notion of the CHSH value of boxes. We then establish Conjecture~\ref{conj1} in the range of parameters $1/\sqrt{2}\leq \eta_1, \eta_2\leq 1$.

\subsection{Structure of the paper} 
This paper is organized as follows. In Section~\ref{sec:HC} we define the hypercontractivity ribbon for bipartite distributions. In Section~\ref{sec:wiringn} wirings of no-signaling boxes is formally defined and some notation is developed for our later use. In Section~\ref{Sec:HC-ribbon} we define a hypercontractivity ribbon for no-signaling boxes, and show that it expands under wirings. In Section~\ref{sec:MC-ribbon} we define our new notion of maximal correlation ribbon and show that it has the tensorization property and is monotone under local operations. The connection between the HC~ribbon and the MC~ribbon is developed in this section.
In Section~\ref{sec:MC-ribbon-wiring} we prove the monotonicity of the MC~ribbon under wirings.
In Section~\ref{sec:isotropic} we study Conjecture~\ref{conj1} for isotropic boxes. Concluding remarks come in Section~\ref{sec:conclusion}. Some of the technical details and proofs are moved to the appendices.

%*****************************************************

%*****************************Hypercontractivity ribbon*************************
\section{Hypercontractivity ribbon} \label{sec:HC}

Let us first fix some notations. Random variables are represented by uppercase letters (such as $A, B$), and we use lowercase letters (such as $a, b$) to denote their values. The alphabet sets of random variables, which throughout this paper are assumed to be finite, are denoted by the calligraphic letters (such as $\mathcal A, \mathcal B$). Then a probability distribution $p_A$ is determined by numbers $p_A(a)$ for $a\in \mathcal A$ which for simplicity is denoted by $p(a)=p_A(a)$.

For natural numbers $k\leq n$ we let $[k:n]=\{k, k+1,\dots, n\}$. We also denote $[n]=[1:n]$.
Moreover for simplicity of notation we use $A_{[n]}=A_1\dots A_n$, and $a_{[n]}=a_1\dots a_n$. 

Entropy of a random variable is defined as $H(A)=-\sum_{a}p(a)\log p(a)$. The conditional entropy is denoted by $H(A|B)$. We have $H(A|B)=H(AB)-H(B)$. Moreover, the condition mutual information is $I(A;B|C)=H(A|C)-H(A|BC)$. By the chain rule we have 
$$I(A;BD|C)=I(A;B|C)+I(A;D|BC).$$ 
We know that $H(A|B)$ and $I(A;B|C)$ are both non-negative. Then for instance if $I(A;BD|C)$ vanishes, then both $I(A;B|C)$ and $I(A;D|BC)$ vanish too.
We will also use the notation
\begin{align}I(A;B;C|D) &= I(A;B|D)-I(A;B|CD)\label{def:three-partite}
\\&=I(A;C|D)-I(A;C|BD)\nonumber
\\&=I(B;C|D)-I(B;C|AD)\nonumber
\end{align}

Let $A,B$ be two random variables with joint distribution $p_{AB}$ that take values in finite sets. Below we define\footnote{We gave a different definition for the hypercontractivity ribbon in the introduction. Later we will comment on the equivalence of these two definitions.} the hypercontractivity ribbon (hereafter, HC~ribbon) associated to  $p_{AB}$. 

\begin{definition}\label{def1}
The hypercontractivity ribbon of $p_{AB}$ denoted by  
$\mathfrak R(A, B)$
is the set of pairs of non-negative numbers $(\lambda_1, \lambda_2)$ such that for every conditional distribution  $p_{U|AB}$
 we have
\begin{align}\label{eq:defHCR}
 \lambda_1I(U; A) + \lambda_2I(U; B) \leq I(U; AB).
\end{align}
\end{definition}

Letting $U=A$ we observe that if $(\lambda_1, \lambda_2)\in \mathfrak R(A, B)$ then $\lambda_1\leq 1$. We similarly have $\lambda_2\leq 1$.
Therefore 
$$\fR(A, B)\subseteq [0,1]^2.$$ 
Furthermore, by data processing inequality $I(U;A), I(U;B)\leq I(U; AB)$. Then $\fR(A, B)$ includes any $(\lambda_1, \lambda_2)$ satisfying 
$0\leq \lambda_1, \lambda_2\leq 1$ and $\lambda_1+\lambda_2\leq 1$. 

The HC~ribbon is equal to  $[0,1]^2$ if and only if $A, B$ are independent. If $(1,1)\in \fR(A, B)$ then by setting $U=AB$ we find that $H(A)+H(B) \leq H(AB)$. Then by the subadditivity inequality $A, B$ are independent. On the other hand, for independent $A, B$ we have  $H(A)+ H(B) = H(AB)$ and 
 $$H(A|U) + H(B|U)\geq H(AB|U),$$
which give \eqref{eq:defHCR} for $(\lambda_1, \lambda_2)=(1, 1)$.

\begin{theorem}\label{thm:ribbon-tensor-data} The HC~ribbon has the following properties:
\begin{enumerate}
\item[\rm{(i)}]\emph{[Tensorization]} If $p_{A_1A_2B_1B_2}=p_{A_1B_1}\cdot p_{A_2B_2}$, then $$\mathfrak R(A_1A_2, B_1B_2)=\mathfrak R(A_1,B_1)\cap\mathfrak R(A_2, B_2).$$

\item[\rm{(ii)}]\emph{[Data processing]} If $p_{A_1A_2B_1B_2}=p_{A_1B_1}\cdot p_{A_2|A_1}\cdot p_{B_2|B_1}$, then 
$$ \mathfrak R(A_1,B_1)\subseteq \mathfrak R(A_2, B_2).$$
\end{enumerate}
\end{theorem}

Part (i) in particular implies that letting $A_iB_i$, $i=1, \dots, n$, be $n$ i.i.d.\ copies of $AB$ then 
$$\fR(A_{[n]},B_{[n]}) = \fR(A , B).$$
Part (ii)
means local transformations on individual random variables can only expand the HC~ribbon. Equivalently, (ii) states that more correlated distributions $p_{AB}$ should have smaller HC~ribbons. 
This is in line with the fact that HC~ribbon is the whole $[0,1]^2$ for independent random variables. On the other hand, as discussed above we always have
\begin{align}\label{eq:hc-ribbon-extreme}
\{(\lambda_1, \lambda_2)\in [0,1]^2\big|\, \lambda_1+\lambda_2\leq 1 \}\subseteq \fR(A, B).
\end{align}
Thus we expect that equality holds for highly correlated distributions $p_{AB}$. Indeed we know that the above inclusion is an equality if and only if $A$ and $B$ have a \emph{common data} (see e.g.,~\cite{KamathAnantharam} and references therein). For example if $A, B$ are binary random variables, and $p(00), p(11)>0$ and $p(01)=p(10)=0$, then we have equality in~\eqref{eq:hc-ribbon-extreme}. Similarly, if $p(01), p(10)>0$ and $p(00)=p(11)=0$, then again equality holds in~\eqref{eq:hc-ribbon-extreme}.

\begin{proof}
(i) For an arbitrary $p_{U|A_1B_1}$ we may define a joint distribution $p_{UA_1A_2B_1B_2}$ by 
\begin{align}\label{eq:ua1b1a2b1-independent-10}
p_{UA_1A_2B_1B_2}=p_{U|A_1B_1}\cdot p_{A_1B_1}\cdot p_{A_2B_2}.
\end{align}
We then have $I(U; A_1A_2B_1B_2)=I(U; A_1B_1)$. Now suppose that $(\lambda_1, \lambda_2)\in \mathfrak R(A_1A_2, B_1B_2)$. 
Thus
\begin{align}
\lambda_1 I(U;A_1) + \lambda_2I(U; B_1)&\leq \lambda_1 I(U;A_1A_2) + \lambda_2 I(U; B_1B_2)\nonumber\\
&\leq I(U; A_1A_2B_1B_2)\nonumber\\
&=I(U; A_1B_1).\label{eq:ua1b1a2b2-9}
\end{align}
Therefore, $(\lambda_1, \lambda_2)\in \mathfrak R(A_1,B_1)$. We similarly have 
$(\lambda_1, \lambda_2)\in \mathfrak R(A_2, B_2)$, and then
$\fR(A_1A_2, B_1B_2)\subseteq \mathfrak R(A_1,B_1)\cap\mathfrak R(A_2,B_2).$

To show the other inclusion let $(\lambda_1, \lambda_2)\in \mathfrak R(A_1,B_1)\cap\mathfrak R(A_2, B_2)$. Take some arbitrary $p_{U|A_1A_2B_1B_2}$. Then we have
\begin{align}
\lambda_1 I(U;A_1) + \lambda_2 I(U; B_1)\leq I(U;A_1B_1),
\label{eqnab1}
\end{align} 
and
\begin{align}
\lambda_1 I(UA_1B_1;A_2) + \lambda_2I(UA_1B_1; B_2)\leq I(UA_1B_1;A_2B_2).
\label{eqnab2}
\end{align}
Observe that
\begin{align*}I(UA_1B_1;A_2)&\geq I(UA_1;A_2)
\\&= I(U;A_2|A_1)
\\&=I(U;A_1A_2)-I(U;A_1),\end{align*}
and similarly $I(UA_1B_1;B_2)\geq I(U;B_1B_2)-I(U;B_1).$
We also have
\begin{align*} 
I(UA_1B_1;A_2B_2)&= I(U;A_2B_2|A_1B_1)
\\&= I(U;A_1A_2B_1B_2)-I(U;A_1B_1).
\end{align*}
Hence, from \eqref{eqnab1} and \eqref{eqnab2} we obtain
\begin{align}
\lambda_1 I(U;A_1A_2) + \lambda_2 I(U; B_1B_2)&\leq I(UA_1B_1;A_2B_2)+I(U; A_1B_1)\nonumber
\nonumber\\
&= I(U;A_1A_2B_1B_2).\label{eqnab33}
\end{align}
Therefore, $(\lambda_1, \lambda_2)\in\mathfrak R(A_1A_2,B_1B_2)$, and 
$$\mathfrak R(A_1,B_1)\cap\mathfrak R(A_2,B_2)\subseteq \fR(A_1A_2,B_1B_2).$$\\

\noindent
(ii) 
By repeated use of the functional representation lemma \cite[Appendix B]{elgamal}, there are random variables $F, G$ that are independent of each other and of $(A_1,B_1)$ such that $A_2$ is a function of $(A_1, F)$, and $B_2$ is a function of $(B_1, G)$. Indeed, $F$ and $G$ can be thought of as the randomness of the channels $p_{A_2|A_1}$ and $p_{B_2|B_1}$. Since $F, G$ are independent, as we discussed earlier $\mathfrak R(F,G)$ is the whole $[0,1]^2$. Therefore by part (i) we have
$$\mathfrak R(A_1F,B_1G)=\mathfrak R(A_1,B_1)\cap \mathfrak R(F, G)=\mathfrak R(A_1,B_1).$$ 
Thus without loss of generality we may assume that the randomness $F, G$ are parts of $A_1, B_1$ respectively, and that $A_2, B_2$ are functions of $A_1, B_1$ respectively. 

Suppose that $(\lambda_1, \lambda_2)\in \mathfrak R(A_1,B_1)$. We have a joint distribution 
$$p_{A_1A_2B_1B_2}=p_{A_1B_1}\cdot p_{A_2|A_1}\cdot p_{B_2|B_1}=p_{A_2B_2} \cdot p_{A_1B_1|A_2B_2}.$$ 
Take some $p_{U|A_2B_2}$. Define 
\begin{align*}p_{UA_1A_2B_1B_2}&:=p_{U|A_2B_2}\cdot p_{A_2B_2}\cdot p_{A_1B_1|A_2B_2}\\&=p_{UA_2B_2}\cdot p_{A_1B_1|A_2B_2}.\end{align*}
Note that the marginal distribution of $p_{UA_1A_2B_1B_2}$ on variables $A_1, A_2, B_1, B_2$ coincides with $p_{A_1A_2B_1B_2}$  that we started with, and 
$I(U; A_1B_1|A_2B_2)=0.$
Therefore,
\begin{align*}
I(U; A_2B_2)& =I(U;A_1A_2B_1B_2)\\
&\geq I(U; A_1B_1)\\
&\geq \lambda_1I(U; A_1) +\lambda_2I(U; B_1)\\
&\geq \lambda_1 I(U; A_2) + \lambda_2 I(U; B_2),
\end{align*}
where in the last line we use the fact that $A_2, B_2$ are functions of $A_1, B_1$ respectively. We are done.

\end{proof}

The standard definition of HC~ribbon~\cite{AhlswedeGacs}, as discussed in the introduction, is in terms of Schatten norms of functions of random variables, rather than mutual information. A remarkable recent work by Nair~\cite{Nair} finds a representation of the HC~ribbon for two random variables in terms of mutual information (that then corresponds to our
definition in the introduction). 

\begin{theorem}[\cite{Nair}] $(\lambda_1, \lambda_2)\in \fR(A, B)$ if and only if for every pair of functions $f_A:\mathcal A\rightarrow \mathbb R$ and $g_B:\mathcal B\rightarrow \mathbb R$ we have
\begin{align}\label{eq:fg-norm-2}
\E[f_Ag_B]\leq \|f_A\|_{\frac{1}{\lambda_1}}\|g_B\|_{\frac{1}{\lambda_2}},
\end{align}
where the Schatten norms are defined by $\|f_A\|_{\frac{1}{\lambda_1}}=\E\big[|f_A|^{1/\lambda_1}\big]^{\lambda_1}$ and similarly for $\|g_B\|_{\frac{1}{\lambda_2}}$.
\label{thm:nair}
\end{theorem}

The following corollary is an immediate consequence of the above theorem and the definition of $\fR(A, B)$ given in Definition~\ref{def1}. This corollary can be directly proved using the Riesz-Thorin theorem (see \cite[Theorem 14]{DelgoshaBeigi}).

\begin{corollary} For every $p_{AB}$ the set of points $(\lambda_1, \lambda_2)\in [0,1]^2$ satisfying~\eqref{eq:fg-norm-2} for every functions $f_A, g_B$, is convex.
\end{corollary}

\subsection{A geometric interpretation of the HC~ribbon}
In Appendix~\ref{app:GW} we discuss a new connection between the HC~ribbon and the Gray-Wyner problem~\cite{GrayWyner} which provides an operational interpretation of the HC~ribbon. Here we briefly discuss a geometric interpretation of the HC~ribbon which will be used in the following sections. 

For every distribution $q_{AB}$ on $\mathcal A\times \mathcal B$ define 
\begin{align}\label{eq:upsilon-def}
\Upsilon(q_{AB}) = \lambda_1H(q_A) + \lambda_2 H(q_B) - H(q_{AB}),
\end{align}
where $H(\cdot)$ is the entropy function. Also, let $\widetilde \Upsilon$ be the point-wise largest function that is convex and $\widetilde \Upsilon(q_{AB})\leq \Upsilon(q_{AB})$ for every distribution $q_{AB}$. 
The function $\widetilde \Upsilon$ is sometimes called the \emph{lower convex envelope} of $\Upsilon$. The following lemma is based on known connections between lower convex envelopes and auxiliary random variables (see~\cite{nairup} for more applications).

\begin{lemma}\label{lem:geom-interpret-HC-ribbon}
For every distribution $p_{AB}$, we have $(\lambda_1, \lambda_2)\in \fR(A, B)$ if and only if $\Upsilon(p_{AB})=\widetilde\Upsilon(p_{AB})$.
\end{lemma}
\begin{proof}
For a given $p_{U|AB}$, by the convexity of $\tilde \Upsilon$ we have
\begin{align*}
\E_U[\Upsilon(p_{AB|U})] &\geq \E_U[\widetilde\Upsilon(p_{AB|U})] \\&\geq \widetilde \Upsilon(\E_U[p_{AB|U}]) \\&= \widetilde \Upsilon(p_{AB}).
\end{align*} 
Then $\widetilde \Upsilon(p_{AB})= \Upsilon(p_{AB})$ implies  
\begin{align}\label{eq:upsilon-pabu-0}
\E_U[\Upsilon(p_{AB|U})] \geq \Upsilon(p_{AB}),
\end{align}
which is equivalent to 
$$I(U;AB)\geq \lambda_1I(U; A) + \lambda_2I(U; B).$$
Therefore, $(\lambda_1, \lambda_2)\in \fR(A,B)$. 

Conversely, $(\lambda_1, \lambda_2)\in \fR(A,B)$ implies that~\eqref{eq:upsilon-pabu-0} holds for every $p_{U|AB}$, and one can verify that this gives $\Upsilon(p_{AB})=\widetilde \Upsilon(p_{AB})$.
\end{proof}

%*****************************Wirings*************************
\section{Wirings of no-signaling boxes} \label{sec:wiringn}
As discussed in the introduction, a non-local box (correlation) is a collection of conditional distributions $p(ab|xy)$. Here $x, y$ are the inputs of the box and $a, b$ are its outputs and $p(ab|xy)$ is the probability of obtaining these outputs.
This quantity $p(ab|xy)$ can be thought of as the probability of obtaining outcomes $a,b$ when we measure subsystems of a bipartite physical system with measurement settings $x, y$ respectively. 
A box has the no-signaling condition if we have
$$p(a|xy) = p(a|x), \qquad \qquad p(b|xy)= p(b|y).$$
That is, the marginal distribution $p(a|xy)$ is independent of $y$, and the marginal distribution $p(b|xy)$ is independent of $x$. Hereafter all the boxes in this paper are assumed to have the no-signaling condition.

Suppose that two parties, say Alice and Bob, are provided with $n$ no-signaling boxes. We denote the inputs of the $i$-th box by $X_i, Y_i$ and its outputs by $A_i, B_i$. Then the $i$-th box is determined by a no-signaling correlation $p_i(a_ib_i|x_iy_i)$. As before, we may think of these $n$ boxes as $n$ \emph{independent} bipartite physical systems whose first subsystem is given to Alice and whose second subsystem is given to Bob. Then each party has $n$ subsystems in hand, and may measure these subsystems in some arbitrary order (independent of the order of the other party).  Each party may choose the input of a box as a (probably random) function of the inputs and outputs of the previous boxes. In fact, the box that is going to be used in each step could itself be chosen as a function of previous inputs and outputs.
With this process the parties end up with a new no-signaling box. Such a process is called a wiring. An example of wirings is shown in Fig.~\ref{fig2}.

Let us describe wirings in a more formal way.  Here we assume that the two parties do not have access to common randomness.  Suppose that two parties want to use the above $n$ boxes as a resource to simulate another box $p(a'b'|x'y')$. Thus Alice is given $x'$ and is asked to output $a'$, and Bob is given $y'$ and is asked to output $b'$ whose  joint distribution is $p(a'b'|x' y')$. 
Alice is going to use the boxes in some order which as explained above can be random itself. Let us denote the corresponding random variables by $\Pi_1, \dots, \Pi_n$. That is, $(\Pi_1, \dots, \Pi_n)$ is a random permutation of $[n]$, and Alice uses box $i$ in her $\Pi_i$-th action. Let us denote the inverse permutation of $(\Pi_1, \dots, \Pi_n)$ by $(\widetilde \Pi_1, \dots, \widetilde \Pi_n)$, i.e.,
$$\wPi_{\Pi_i}=i,  \qquad \Pi_{\wPi_i}=i. $$
Then Alice first uses box $\wPi_1$, and then uses box $\wPi_2$ and so on.

Now let us describe Alice's $j$-th action. Before the $j$-th action Alice has used boxes $\wPi_1=\wpi_1, \dots, \wPi_{j-1}=\wpi_{j-1}$ with inputs $X_{\wPi_1}=x_{\wpi_1}, \dots, X_{\wPi_{j-1}}=x_{\wpi_{j-1}}$ and has observed outputs $A_{\wPi_1}=a_{\wpi_1}, \dots, A_{\wPi_{j-1}}=a_{\wpi_{j-1}}$. To simplify our notation let us define
$$\widetilde X_i:= X_{\wPi_i}, \qquad \qquad \widetilde A_i:=A_{\wPi_i}.$$
That is $\widetilde X_i, \widetilde A_i$ are the input and output of the box that Alice uses in her $i$-th action. By this notation before her $j$-th action Alice has used boxes $\wPi_{[j-1]}=\wpi_{[j-1]}$ with inputs $\widetilde X_{[j-1]}=\tilde x_{[j-1]}$ and has observed outputs $\widetilde A_{[j-1]} = \tilde a_{[j-1]}$. She also has $x'$ from the beginning. Then she chooses the next box and its input according to some stochastic map 
\begin{align}\label{eq:encoding-alice}
q\big(\wpi_{j}\tilde x_j\,\big|\,\wpi_{[j-1]}\tilde a_{[j-1]}\tilde x_{[j-1]} x'\big).
\end{align}
She puts $\tilde x_{j}= x_{\wpi_j}$ as the input of box $\wpi_j$ and observes $\tilde a_j=a_{\wpi_j}$ as the output. She continues until using all the $n$ boxes. A summary of the definition of the random variables defined above is given in Table \ref{table:summary}.

The actions of Bob are described similarly. We denote the random order under which Bob uses the boxes by $\wOmega_1, \dots, \wOmega_n$ and its inverse permutation by $ \Omega_1, \dots, \Omega_n$, i.e., 
$$\Omega_{\widetilde \Omega_i}=i, \qquad\qquad \wOmega_{ \Omega_i}=i,$$
and Bob uses box $i$ in his $\Omega_i$-th action.
We use
$$\widetilde Y_i:= Y_{\wOmega_i}, \qquad \qquad \widetilde B_i:=B_{\wOmega_i}.$$
Then before his $j$-th action, Bob has used boxes $\wOmega_{[j-1]}=\womega_{[j-1]}$, with inputs $\widetilde Y_{[j-1]}=\wy_{[j-1]}$ and has observed outputs $\widetilde B_{[j-1]}=\wb_{[j-1]}$. He also has $y'$ from the beginning. Then he uses some stochastic map
\begin{align}\label{eq:encoding-bob}
q\big(\womega_{j}\wy_j\,\big|\,\womega_{[j-1]}\wb_{[j-1]}\wy_{[j-1]} y'\big),
\end{align}
to choose $\wOmega_j=\womega_j$ and $\wy_j=y_{\womega_j}$. He puts $y_{\womega_j}$ as the input of box $\womega_j$ and receives output $\wb_j=b_{\wpi_j}$. He continues until using all the boxes. 

In~\eqref{eq:encoding-alice} and \eqref{eq:encoding-bob} we use $q(\cdot| \cdot)$ for stochastic maps of both Alice and Bob. This however should not cause any confusion since whether $q(\cdot | \cdot)$ corresponds to Alice or Bob's action should be clear from its arguments. 

We need to simplify our notation even further. Let us denote $T_i$ be the \emph{transcript} of Alice (whatever she has) before using the $i$-th box (before her $\wPi_i$-th action), i.e.,
\begin{align*}T_i&:=\wPi_1\dots \wPi_{\Pi_i-1} X_{\wPi_1}\dots X_{\wPi_{\Pi_i-1}} A_{\wPi_1}\dots A_{\wPi_{\Pi_i-1}}\\&=\wPi_{[\Pi_i-1]}\wX_{[\Pi_i-1]} \wA_{[\Pi_i-1]}.\end{align*}
We also use $\wT_i$ for the transcript of Alice before her $i$-th action, i.e.,
\begin{align*}\wT_i&:=T_{\wPi_i} = \wPi_1\dots \wPi_{i-1} X_{\wPi_1}\dots X_{\wPi_{i-1}} A_{\wPi_1}\dots A_{\wPi_{i-1}}\\&=\wPi_{[i-1]}\wX_{[i-1]}\wA_{[i-1]}.\end{align*}
We define $S_i$ and $\wS_i$ similarly for Bob, i.e.,
$$S_i:= \wOmega_{[\Omega_i-1]}\wY_{[\Omega_i-1]} \wB_{[\Omega_i-1]},  \qquad \wS_i:=\wOmega_{[i-1]}\wY_{[i-1]}\wB_{[i-1]}.$$
With these notations Alice before using the $i$-th box has $T_i=t_i$ and $x'$ in hand and with probability $q(i x_i|t_i x')$ chooses the $i$-th box for her next action and puts $x_i$ in this box. Similarly before using the $i$-th box, Bob has $s_i$ and $y'$ and with probability $q(i y_i|s_i y')$ chooses box $i$ for his next action and puts $y_i$ as its input. As a result, the joint probability of inputs and outputs of the boxes and the orderings of Alice and Bob is 
\begin{align}\label{eq:p-a-x-m-n-y}
&p\big(a_{[n]}b_{[n]}x_{[n]}y_{[n]} \pi_{[n]}\omega_{[n]}\,\big|\,x'y'\big)=\prod_{i=1}^n\bigg[ p_i\big(a_ib_i\,\big|\,x_iy_i\big) q\big(ix_i\,\big|\,t_i x'\big)q\big(iy_i\,\big|\,s_i y'\big)\bigg].
\end{align}
An extended discussion of how to arrive at this form in the above equation can be found  in the beginning of Appendix~\ref{app:lemm-proof}.

At the end of wirings Alice applies the stochastic map $q(a'|a_{[n]}x_{[n]}\pi_{[n]} x')$ to determine her final output and Bob applies $q(b'|b_{[n]}y_{[n]} \omega_{[n]}y')$ to determine his final output.

\begin{table*}
\begin{center}
\begin{small}
\begin{tabular}{|c|c|c|}
\hline
Notation & Description   & { Corresponding variable of Bob}\\
\hline
$\Pi_i$ & Alice uses the $i$-th box in her $\Pi_i$-th action& $\Omega_i$\\ \hline   
 & Index of the box Alice uses in her $i$-th action: & \\
$\wPi_i$&$~\Pi_{\wPi_i}=i, \qquad \wPi_{\Pi_i}=i$&$\wOmega_i$\\
\hline
$X_i$ & Alice's input of the $i$-th box & $Y_i$\\\hline
$A_i$ & Alice's output of the $i$-th box & $B_i$\\\hline
& Alice's input in her $i$-th action:  & \\
$\wX_i$ &$~\wX_i=X_{\wPi_i}$&$\wY_i$\\
\hline
 & Alice's output in her $i$-th action: & \\
$\wA_i$&$~\wA_i=A_{\wPi_i}$&$\wB_i$\\
\hline
$T_i$ & Alice's transcript before using the $i$-th box & $S_i$\\\hline
 & Alice's transcript before her $i$-th action: & \\
$\wT_i$&$~\wT_i=T_{\wPi_i}$& $\wS_i$\\
\hline
$T_i^e$ & $T_iX_i\Pi_i$ & $S_i^e$
\\\hline
\end{tabular}
\end{small}
\end{center}
\caption{Alice and Bob use the $n$ no-signaling boxes in different (probably random) orders.  This table is a summary of notations used to describe the random variables associated with these orders as well as the inputs and the outputs of the boxes. Here by Alice's transcript we mean whatever Alice has observed up to a certain point. 
}
\label{table:summary}
\end{table*}%

\begin{lemma}\label{lem:dist-wiring} For any given $x', y'$ the followings hold.
\begin{enumerate}
\item[\rm{(i)}] $I\big(A_iB_i; T_iS_i\Pi_i\Omega_i\big|X_iY_i, x'y'\big)=0$.

\item[\rm{(ii)}] $I\big(A_i; S_i Y_i \Omega_i\big|\, T_i X_i \Pi_i, x'y'\big)=0$ and $I\big(B_i; T_i X_i \Pi_i\big|\, S_i Y_i \Omega_i, x'y'\big)=0$.

\item[\rm{(iii)}] $I\big(A_i;  B_{[n]}Y_{[n]}\Omega_{[n]}  \big|\, T_iX_i \Pi_i B_i Y_i\Omega_i,x'y'\big)=0$ and $I\big(B_i;  A_{[n]}X_{[n]}\Pi_{[n]}  \big|\, T_iA_iX_i\Pi_i S_i  Y_i\Omega_i,x'y'\big)=0.$

\item[\rm{(iv)}]  $I\big(\wX_i\wPi_i; B_{[n]}Y_{[n]}\Omega_{[n]}\big|\, \wT_i,x'y'\big)=0$ and  $I\big(\wY_i\wOmega_i; A_{[n]}X_{[n]}\Pi_{[n]}\big|\, \wS_i,x'y'\big)=0$.
\end{enumerate}
\end{lemma}

Here we give an informal intuitive proof of this lemma. For a full detailed proof see Appendix~\ref{app:lemm-proof}.

\begin{proof}[Informal proof] (i) holds simply because given the inputs of the $i$-th box, its outputs are independent of the transcripts  of Alice and Bob when they reach this box. 
(ii) is a consequence of (i) and the no-signaling condition.
(iii) holds because when (say) Alice uses the $i$-th box, her output, if not conditioned on her future observations, depends only on the inputs of the $i$-th box and Bob's output of this box.
(iv) is a simple consequence of the fact that $\wX_i\wPi_i$ and $\wY_i\wOmega_i$ are generated locally without using the boxes.

\end{proof}

We will frequently use the following lemma.
\begin{lemma}\label{lemma:amnm} For auxiliary random variables $U$ and $V$ we have
\begin{align}
&I\big(U;A_{[n]}X_{[n]}\Pi_{[n]}|V, x'y'\big)=\sum_{i=1}^n\bigg[I\big(U;\wX_{i}\wPi_i|\wT_iV, x'y'\big)+I\big(U;A_{i}|T^e_iV, x'y'\big)\bigg], \label{eq:tyab5}
\\&H\big(A_{[n]}X_{[n]}\Pi_{[n]}|V, x'y'\big)=\sum_{i=1}^n\bigg[H\big(\wX_{i}\wPi_i|\wT_iV, x'y'\big)+H\big(A_{i}|T^e_iV, x'y'\big)\bigg],\label{eq:tyab5new}
\end{align}
where $T_i^{e}:=T_i X_{i} \Pi_{i}.$ Similar equations hold for Bob's random variables too.
\end{lemma}
This lemma follows from repeated use of the chain rule for conditional mutual information and its proof is given in Appendix~\ref{app:lemma-amnnproof}.

%*************************HC~ribbon & wiring***************************************
\section{HC~ribbon for no-signaling boxes}\label{Sec:HC-ribbon}

In this section we define the HC~ribbon for no-signaling boxes, and show that it is well-behaved under wirings.

\begin{definition}\label{definition:HC-no-signaling}
Given a no-signaling box $p(ab|xy)$, we define its HC~ribbon to be the intersection of the HC~ribbons of its outputs conditioned on all possible inputs, i.e.,
$$\mathfrak R(A, B|X, Y):=\bigcap_{x,y}\mathfrak R(A, B|X=x, Y=y).$$
\end{definition}

Let us as an example, compute the HC~ribbon of the perfect PR box (which we denoted by $\PR_1$). For any $x,y\in \{0,1\}$, $\Pr_1(a, b|x,y)=1/2$ iff $a\oplus b=xy$. Then by the discussion before the proof of Theorem~\ref{thm:ribbon-tensor-data}, we have that for any $x,y
\in\{0,1\}$,
$$\fR(\Pr_1(a, b|x, y)) = \big\{(\lambda_1, \lambda_2)\in [0,1]^2\big|\, \lambda_1+\lambda_2\leq 1\big\}.$$
As a result $\fR(\Pr_1)$ which is the intersection of the above four HC~ribbons, is equal to 
\begin{align}\label{eq:ribbon-perfect-pr}
\fR(\Pr_1)=\big\{(\lambda_1, \lambda_2)\in [0,1]^2\big|\, \lambda_1+\lambda_2\leq 1\big\}.
\end{align}

We can now state the main theorem of this section.

\begin{theorem}\label{thm:main-theorem}
Suppose that a no-signaling box $p(a'b'| x'y')$ can be generated from $n$ no-signaling boxes $p_i(a_ib_i| x_i y_i)$ where $i\in [n]$, under wirings. Then we have
\begin{align}\label{eq:mon-wiring-ribbon}
\bigcap_{i=1}^n \fR(A_i,B_i|X_i, Y_i) \subseteq \fR(A',B'|X',Y').
\end{align}
\end{theorem}

Observe that this theorem is consistent with the known protocols for non-locality distillation with wirings. For example in~\cite{BS09} it is shown that using certain no-signaling boxes one can simulate the perfect PR box under wirings. Nevertheless, it can be verified that the HC~ribbons of those boxes is equal to the HC~ribbon of the perfect PR box computed in~\eqref{eq:ribbon-perfect-pr}.

In the following proof for wirings of no-signaling boxes we use the notation developed in the previous section.

\begin{proof} By definition we need to show that 
$$\bigcap_{i=1}^n \fR(A_i, B_i|X_i, Y_i) \subseteq \fR(A', B'|X'=x', Y'=y'),$$
for every $x',y'$. So we fix $x', y'$ and in the following for simplicity of notation drop all conditionings on $x', y'$.

Let $(\lambda_1, \lambda_2)$ be in $\fR(A_i, B_i|X_i, Y_i)$ for all $i\in [n]$. We need to show that $(\lambda_1, \lambda_2)\in \fR(A', B'|X'=x', Y'=y')$.  As explained in Section~\ref{sec:HC}, any HC~ribbon always includes pairs $(\lambda_1, \lambda_2)$ that satisfy  $\lambda_1+\lambda_2\leq 1$. Therefore if $\lambda_1+\lambda_2\leq 1$, there is nothing left to prove. So in the following we assume that $\lambda_1, \lambda_2\in [0,1]$ are such that 
$$\lambda_1+\lambda_2\geq 1.$$

Recall that $A', B'$ are generated by Alice and Bob under local stochastic maps. That is, Alice generates $A'$ given 
$A_{[n]}X_{[n]}\Pi_{[n]}$ and Bob generates $B'$ given 
$B_{[n]}Y_{[n]} \Omega_{[n]}$. Therefore by part (ii) of Theorem~\ref{thm:ribbon-tensor-data} (data processing for HC~ribbon) we only need to prove 
\begin{align}\label{eq:sum-leq-1}
(\lambda_1, \lambda_2)\in \fR(A_{[n]}X_{[n]}\Pi_{[n]}, B_{[n]}Y_{[n]} \Omega_{[n]}).
\end{align}
Note that the HC~ribbon on the right hand side is computed for the distribution induced by the wirings of boxes, i.e., with respect to distribution \eqref{eq:p-a-x-m-n-y}.

Let $U$ be an auxiliary random variable determined by $p_{U|A_{[n]}X_{[n]}\Pi_{[n]}B_{[n]}Y_{[n]} \Omega_{[n]}}$. We would like to show that 
\begin{align}
&\lambda_1 I\big(U;A_{[n]}X_{[n]}\Pi_{[n]}\big) + \lambda_2 I(U; B_{[n]}Y_{[n]} \Omega_{[n]})\leq I\big(U;A_{[n]}X_{[n]}\Pi_{[n]}B_{[n]}Y_{[n]} \Omega_{[n]}\big).\label{eq:final-ineq-3}
\end{align}

Using the first equation of Lemma \ref{lemma:amnm} with $V=\emptyset$, we get that
\begin{align}
I\big(U;A_{[n]}X_{[n]}\Pi_{[n]}\big)&=\sum_{i=1}^n\bigg[I\big(U;\wX_{i}\wPi_i|\wT_i\big)+I\big(U;A_{i}|T^e_i\big)\bigg]\nonumber\\
&= \sum_{i=1}^n\bigg[I\big(U;\wX_{i}\wPi_i\big|\wT_i\big)+I\big(U;A_{i}\big|T^e_iS^e_i\big) + I\big(U;A_{i};S^e_{i}\big|T_i^e\big)\bigg].\label{eq:123a}
\end{align}
where in the last step, we used the definition given in \eqref{def:three-partite}.
We similarly have
\begin{align}\label{eq:123b}
&I\big(U;B_{[n]}Y_{[n]}\Omega_{[n]}\big)
= \sum_{i=1}^n\bigg[I\big(U;\wY_{i}\wOmega_i\big|\wS_i\big)+I\big(U;B_{i}\big|T^e_iS^e_i\big) + I\big(U;B_{i};T^e_{i}\big|\,S_i^e\big)\bigg]
\end{align}
where $$S_i^{e}:=S_iY_i\Omega_i.$$
From  $(\lambda_1, \lambda_2)\in \fR(A_i, B_i|X_i, Y_i)$ we have
\begin{align}
 &\lambda_1I\big(UT_{i}S_i\Pi_i\Omega_i;A_{i}\big|\,X_{i}Y_i\big)+  \lambda_2I\big(UT_{i}S_i\Pi_i\Omega_i;B_{i}\big|\,X_{i}Y_i\big)\leq I\big(UT_{i}S_i\Pi_i\Omega_i;A_iB_{i}\big|\,X_{i} Y_i\big).\label{eq:mc-12-lambda}
\end{align}
Indeed by definition this inequality holds for every $(X_i, Y_i)=(x_i, y_i)$, and then holds for their average. On the other hand by Lemma \ref{lem:dist-wiring} part (i) we have $I\big(A_iB_i; T_iS_i\Pi_i\Omega_i\big|X_iY_i\big)=0$. Thus an application of chain rule gives
\begin{align*}
 \lambda_1I\big(U;A_{i}\big|\,T^e_{i} S^e_i\big)+  \lambda_2I\big(U;B_{i}\big|\,T^e_{i} S^e_i\big)&\leq I\big(U;A_iB_{i}\big|\,T^e_{i} S^e_i\big).
\end{align*}
Therefore using \eqref{eq:123a} and \eqref{eq:123b}, to prove \eqref{eq:final-ineq-3} we need to show that 
 $\chi(\lambda_1, \lambda_2)\geq 0$ for any $\lambda_1, \lambda_2\in[0,1]$ satisfying $\lambda_1+\lambda_2\geq 1$, where
\begin{align*}
\chi(\lambda_1, &\lambda_2):=-\sum_{i=1}^n \bigg[\lambda_{1}I\big(U;\wX_{i}\wPi_i|\wT_i\big)+\lambda_2I\big(U;\wY_{i}\wOmega_i|\wS_i\big)+\lambda_1I\big(U;A_{i};S^e_{i}|T_i^e\big) \\
&\qquad\qquad\qquad ~+ \lambda_2 I\big(U;B_{i};T^e_{i}\big|\,S_i^e\big) +I\big(U;A_iB_{i}\big|\,T^e_{i} S^e_i\big)\bigg]
 +I\big(U;A_{[n]}X_{[n]}\Pi_{[n]}B_{[n]}Y_{[n]}\Omega_{[n]}\big).
\end{align*}
In Appendix~\ref{app:proof-main-thm} using chain rule we will first find an equivalent expression of $\chi(\lambda_1, \lambda_2)$ and
then using Lemma~\ref{lem:dist-wiring} by a term by term analysis of the expression we show that it is non-negative. 

\end{proof}

The following corollary is a simple consequence of the above theorem.

\begin{corollary}\label{corol:HC-ribbon-close-wiring} Let $\Lambda\subseteq [0,1]^2$ be an arbitrary subset. Then the set of no-signaling boxes whose HC~ribbon contains $\Lambda$ is closed under wirings.
\end{corollary}

Theorem \ref{thm:main-theorem} (on HC ribbon under wiring of no-signaling boxes) can be interpreted as a generalization of the tensorization property of hypercontractivity ribbon (part (i) of Theorem \ref{thm:ribbon-tensor-data}). For example we have the following.

\begin{corollary}  For any four random variables $A_1, B_1, A_2, B_2$ satisfying 
\begin{align}I(A_2;B_1|A_1)=I(A_1;B_2|B_1)=0,\label{eqn:non-signalingconditionAB}\end{align}
we have
$$\mathfrak R(A_1,B_1)\cap \mathfrak R(A_2,B_2|A_1,B_1)\subseteq\mathfrak R(A_1A_2,B_1B_2),$$
where $\mathfrak R(A_2,B_2|A_1,B_1)$ is defined in Definition~\ref{definition:HC-no-signaling}. 
\end{corollary}

Note that this is indeed a generalization of the tensorization property since when $(A_1,B_1)$ is independent of $(A_2,B_2)$, this property reduces to the tensorization property.

\begin{proof}
Consider two bipartite no-signaling boxes as follows. The first box is determined by the conditional distribution
$p_{A_1B_1|X_1Y_1}=p_{A_1B_1}$, and the second box is defined by 
$$p_{A_2B_2|X_2Y_2}:=p_{A_2B_2|A_1B_1},$$ 
where $\mathcal X_2=\mathcal A_1$ and $\mathcal Y_2=\mathcal B_1$.
The first box is obviously no-signaling. The second box is guaranteed to be no-signaling by~\eqref{eqn:non-signalingconditionAB}. Both parties first use the first box, and then the second box by directly wiring the output of the first box to the input of the second box. This allows Alice and Bob to simulate a channel whose input is $(X_1, Y_1)$ and whose output is $(A_1A_2, B_1B_2)$. However $p_{A_1A_2,B_1B_2|X_1,Y_1}=p_{A_1A_2,B_1B_2}$. Then the results follows as a very special case of Theorem~\ref{thm:main-theorem} (on HC ribbon under wiring of no-signaling boxes).
\end{proof}

%********************************MC~ribbon**************************************

\section{Maximal correlation ribbon}\label{sec:MC-ribbon}
In the previous section, based on Corollary~\ref{corol:HC-ribbon-close-wiring}, we obtain a systematic method to construct sets of no-signaling boxes that are closed under wirings. Nevertheless, to construct such sets we need to be able to compute the HC~ribbons of no-signaling boxes, for which we do not know an efficient algorithm. Our goal in this and the following sections is to define another invariant of no-signaling boxes with similar monotonicity properties as the HC~ribbons, that is efficiently computable.

Given a bipartite distribution $p_{AB}$ we consider functions $f_{AB}:\mathcal A\times \mathcal B\rightarrow \mathbb R$. Then we denote its expectation value by $\E[f]$. 
Sometimes we denote $\E[f]$ by $\E_{AB}[f]$ to emphasis that the expectation is computed with respect to the distribution $p_{AB}$. We may also consider the conditional expectation $\E_{A|B}[f]$, and view it as a random variable taking the value $\E_{A|B=b}[f]$ (the expectation of $f_{AB}$ over the conditional distribution $P_{A|B=b}$) whenever $B=b$. In other words, $\E_{A|B}[f]$ is viewed as a function of $B$ which itself is a random variable. 

The variance of $f_{AB}$ is denoted by 
$$\Var[f]:=\E[(f-\E[f])^2]=\E[f^2]-\E[f]^2.$$ 
Again sometimes we denote $\Var[f]$ by $\Var_{AB}[f]$. We also consider the conditional variance $\Var_{A|B}[f]:= \E_{A|B}[(f-\E_{A|B}[f])^2]$ which again is a function of $B$. We will frequently use the \emph{law of total variance} which states that
$$\Var[f]=\Var_A\E_{B|A}[f] + \E_A\Var_{B|A}[f].$$
Since variance is always non-negative, from the law of total variance we find that
\begin{align}\label{eq:total-variance-ineq}
\Var[f]\geq \max\{ \Var_A\E_{B|A}[f], \E_A\Var_{B|A}[f]\}.
\end{align}

Now we are ready to define the maximal correlation ribbon (MC~ribbon) of bipartite distributions.

\begin{definition} The maximal correlation ribbon of $p_{AB}$ denoted by $\fS(A, B)$ is the set of all pairs $(\lambda_1, \lambda_2)$ of non-negative numbers such that for all functions $f_{AB}:\mathcal A\times \mathcal B\rightarrow \mathbb R$ we have
\begin{align}
\Var[f] \geq \lambda_1 \Var_{A}\E_{B|A}[f] + \lambda_2 \Var_{B}\E_{A|B}[f].
\label{eq:def-cond-var-lambda}
\end{align}
\end{definition}

Letting $f$ be a function of $A$ only, by the law of total variance we have
$\Var[f] = \Var_A\E_{B|A}[f]$. Then if $(\lambda_1, \lambda_2)\in \fS(A, B)$, we have $\lambda_1\leq 1$. We similarly have $\lambda_2\leq 1$. Therefore, 
$$\fS(A, B)\subseteq [0,1]^2.$$
Furthermore observe that by~\eqref{eq:total-variance-ineq} for all $\lambda_1, \lambda_2\geq 0$ with $\lambda_1+\lambda_2\leq 1$ we have $(\lambda_1, \lambda_2)\in \fS(A,B)$. 

As an example, let us  compute the MC~ribbon in the case where $A$ and $B$ are independent. Using the fact that $A$ and $B$ are independent and the convexity of $t\mapsto t^2$, (see Lemma~\ref{lem:abc-e-var} in Appendix~\ref{app:MC-ribbon-tensor-data} for details) it can be verified that
\begin{align*} 
\E_{A}\Var_{B|A}[f] \geq \Var_B\E_{A|B}[f].
\end{align*}
Thus using the law of total variance we have
\begin{align*}
\Var[f] &= \Var_{A}\E_{B|A}[f] + \E_{A}\Var_{B|A}[f]
 \\&\geq \Var_{A}\E_{B|A}[f] + \Var_{B}\E_{A|B}[f],
\end{align*}
As a result, $(1, 1)\in \fS(A, B)$ which gives $\fS(A,B)=[0,1]^2$.

Note that for any function $f_{AB}$ we have $f=\tilde f + \E[f]$ for some function $\tilde f$ with $\E[\tilde f]=0$. Then rewriting the definition of MC~ribbon in terms of $\tilde f$ we obtain the following equivalent characterization of the MC~ribbon.

\begin{lemma}\label{lem:def-MC-ribbon-zero}
For a bipartite distribution $p_{AB}$, its MC~ribbon $\fS(A, B)$ is the set of pairs $(\lambda_1, \lambda_2)$ of non-negative numbers such that for every function $f_{AB}$ with $\E[f]=0$ we have
$$\E[f^2]\geq \lambda_1 \E_A[(\E_{B|A}[f])^2] +\lambda_2 \E_B[(\E_{A|B}[f])^2].$$
\end{lemma}

The following theorem states the main properties of the MC~ribbon, and is the analogue of Theorem~\ref{thm:ribbon-tensor-data} (providing the data processing and tensorization of the HC ribbon).

\begin{theorem}\label{thm:MC-ribbon-tensor-data} 
The MC~ribbon has the following properties:
\begin{enumerate}
\item[\rm{(i)}]\emph{[Tensorization]} If $p_{A_1A_2B_1B_2}=p_{A_1B_1}\cdot p_{A_2B_2}$, then $$\mathfrak S(A_1A_2, B_1B_2)=\mathfrak S(A_1,B_1)\cap\mathfrak S(A_2, B_2).$$

\item[\rm{(ii)}]\emph{[Data processing]} If $p_{A_1A_2B_1B_2}=p_{A_1B_1}\cdot p_{A_2|A_1}\cdot p_{B_2|B_1}$, then 
$$ \mathfrak S(A_1,B_1)\subseteq \mathfrak S(A_2, B_2).$$
\end{enumerate}
\end{theorem}

We will argue later that this theorem can be proved as a consequence of Theorem~\ref{thm:ribbon-tensor-data}. However, for the sake of completeness,
in Appendix~\ref{app:MC-ribbon-tensor-data} we present a direct proof for this theorem that is based on the law of total variance.

%**********************************

%********************MC~ribbon & HC~ribbon************************
\subsection{MC~ribbon vs HC~ribbon}\label{subsec:MC-HC-ribbons}

MC~ribbon and HC~ribbon have similar properties. They both are equal to $[0,1]^2$ for independent bipartite distributions, have the tensorization property and satisfy the data processing inequality. Furthermore, the proofs of these results for MC~ribbon are similar to their proofs for HC~ribbon. To prove Theorem~\ref{thm:ribbon-tensor-data} (data processing and tensorization properties of the HC ribbon), our basic tool is the chain rule. Similarly to prove Theorem~\ref{thm:MC-ribbon-tensor-data} (data processing and tensorization properties of the MC ribbon) in Appendix~\ref{app:MC-ribbon-tensor-data} we use the law of total variance which can be thought as a chain rule for variance. In the following we make the connection between these ribbons more precise.  

Recall that in Lemma~\ref{lem:geom-interpret-HC-ribbon} we prove that $(\lambda_1, \lambda_2)$ belongs to $\fR(A, B)$ if $\Upsilon(p_{AB})$ defined by
$$\Upsilon(q_{AB}) = \lambda_1H(q_A) + \lambda_2 H(q_B) - H(q_{AB}),$$
matches its lower convex envelope denoted by $\widetilde \Upsilon$, at $p_{AB}$, i.e, $(\lambda_1, \lambda_2)\in \fR(A, B)$ if and only if $\Upsilon(p_{AB})=\widetilde \Upsilon(p_{AB})$. In particular, this implies that 
$\Upsilon$ is \emph{locally convex} at $p_{AB}$. To make this latter notation more precise we consider the following perturbation around $p_{AB}$.
Given a function $f_{AB}$ with $\E[f]=0$, define
\begin{align}
q_{AB}^{(\epsilon)}:=p_{AB}(1+\epsilon f_{AB}).
\label{sweepsupport}
\end{align}
Then $q_{AB}^{\epsilon}$ is a probability distribution for sufficiently small $|\epsilon|$, and we may consider $g(\epsilon)=\Upsilon(q_{AB}^{(\epsilon)})$. A straightforward calculation~\cite[Lemma 2]{GMartonPaper} verifies that\footnote{Here to exactly get this expression, we should take natural logarithm instead of logarithm in base 2 in the definition of the entropy function.} 
\begin{align*}
g''(0) = \E[f^2] -\lambda_1 \E_A[(\E_{B|A}[f])^2)] -\lambda_2 \E_B[(\E_{A|B}[f])^2.
\end{align*}
Then, according to Lemma~\ref{lem:def-MC-ribbon-zero}, local convexity for this class of perturbations holds, i.e., $g''(0)\geq 0$ for every choice of $f$, if and only if $(\lambda_1, \lambda_2)\in \fS(A, B)$. 

The following theorem states the above observation in the context of auxiliary random variables (see also \cite[Theorem 4]{Anantharametal}). The main ideas of its proof are already discussed, so we leave a detailed proof for Appendix~\ref{app:MC-HC-ribbon-proof}.

\begin{theorem}\label{thm:MC-HC-ribbon}
The followings hold.
\begin{enumerate}
\item[\rm{(i)}] $(\lambda_1, \lambda_2)\in \fS(A,B)$ if and only if there exists a constant $K\geq 0$ such that for all $f_{AB}$ with $\E[f]=0$ and $\Var[f]=1$ we have
$$I(U_{\epsilon}; AB) + K\epsilon^3 \geq \lambda_1 I(U_{\epsilon};A) + \lambda_2 I(U_{\epsilon}; B),$$
where $p_{ABU_\epsilon}$ is defined by $p(U^{\epsilon}=+1)=p(U^{\epsilon}=-1)=1/2$ and 
\begin{align}\label{eq:def-p-ab-u-epsilon}
p_{AB|U_{\epsilon}=u}=p_{AB}(1+\epsilon uf_{AB}).
\end{align}

\item[\rm{(ii)}] $(\lambda_1, \lambda_2)\in \fS(A,B)$ if and only if there exists a constant $K\geq 0$ such that for all $p_{U|AB}$ we have
\begin{align*}&I(U; AB) + K\cdot \E_U\big[\|p_{AB|U} - p_{AB}\|_1^3\big] 
\geq \lambda_1 I(U;A) + \lambda_2 I(U; B),\end{align*}
where $\|\cdot\|_1$ denotes\footnote{Since all norms on a finite dimensional vector space are equivalent, in the statement of the theorem we could replace Schatten one-norm with any other norm.} the norm-$1$.

\item[\rm{(iii)}] $\fR(A, B)\subseteq \fS(A, B)$.
\end{enumerate}
\end{theorem} 

\begin{remark}
For simplicity, we sometimes use the big O notation, replacing $K\epsilon^3$ and $K\cdot \E_U\big[\|p_{AB|U} - p_{AB}\|_1^3\big]$ with $O(\epsilon^3)$ and $O(\E_U\big[\|p_{AB|U} - p_{AB}\|_1^3\big])$ respectively. Throughout, whenever we write $O(\cdot)$ we mean multiplication of a constant that depends only on the underlying distribution and $\lambda_1, \lambda_2$.
\end{remark}

We can now obtain a proof for Theorem~\ref{thm:MC-ribbon-tensor-data} (data processing and tensorization properties of the MC ribbon) using the above theorem. We may follow the same steps as in the proof of Theorem~\ref{thm:ribbon-tensor-data} (data processing and tensorization properties of the HC ribbon) and only take care of the third order correction terms. Here with this idea we present a proof for part (i) of Theorem~\ref{thm:MC-ribbon-tensor-data}. A proof for part (ii) is obtained similarly.

Suppose that $p_{A_1A_2B_1B_2}=p_{A_1B_1}\cdot p_{A_2B_2}$, and assume that $(\lambda_1, \lambda_2)\in \fS(A_1A_2, B_1B_2)$. 
Take an arbitrary $p_{U|A_1B_1}$ and define $p_{A_1A_2B_1B_2U}$ using~\eqref{eq:ua1b1a2b1-independent-10}. Then following the proof of part (i) of Theorem~\ref{thm:ribbon-tensor-data} we should add the extra term 
$$O\big(\E_U\big[\|p_{A_1A_2B_1B_2|U} - p_{A_1A_2B_1B_2}\|_1^3\big]\big),$$ 
to the second line of~\eqref{eq:ua1b1a2b2-9}. Then for the third line of~\eqref{eq:ua1b1a2b2-9} we use  
$$\|p_{A_1A_2B_1B_2|U=u} - p_{A_1A_2B_1B_2}\|_1= \|p_{A_1B_1|U=u} - p_{A_1B_1}\|_1,$$
which is implied by $p_{A_1A_2B_1B_2|u}=p_{A_1B_1|u}\cdot p_{A_2B_2}$ and $p_{A_1A_2B_1B_2}=p_{A_1B_1}\cdot p_{A_2B_2}$.  
Then using the second characterization of the MC~ribbon in the above theorem, we obtain $\fS(A_1A_2, B_1B_2)\subseteq \fS(A_1, B_1)$. We similarly have $\fS(A_1A_2, B_1B_2)\subseteq \fS(A_2, B_2)$.

For the other direction we use the first and second characterizations of the MC~ribbon in the above theorem simultaneously. Fix some $f_{A_1A_2B_1B_2}$ with $\E[f]=0$ and $\Var[f]=\E[f^2]=1$, and define $p_{A_1A_2B_1B_2U_\epsilon}$ as in~\eqref{eq:def-p-ab-u-epsilon}. Then following the proof of part (i) of Theorem~\ref{thm:ribbon-tensor-data} we should add the extra term
$$O(\E_{U_\epsilon}[\|p_{A_1B_1|U_\epsilon} -p_{A_1B_1}\|_1^3]),$$ 
to the right hand side of~\eqref{eqnab1}, and the extra term 
$$O(\E_{U_\epsilon A_1B_1}[\|p_{A_2B_2|U_\epsilon A_1B_1} - p_{A_2B_2}\|_1^3]),$$
to the right hand side of~\eqref{eqnab2}. Then to write down~\eqref{eqnab33} we add up the above two terms and verify that
\begin{align}
&\E_{U_\epsilon}[\|p_{A_1B_1|U_\epsilon} -p_{A_1B_1}\|_1^3] \leq \E_{U_\epsilon}[\|p_{A_1A_2B_1B_2|U_\epsilon} - p_{A_1A_2B_1B_2}\|_1^3] = O(\epsilon^3),\label{eq:ineq-a1b1ua2b2-5}
\end{align}
and
\begin{align}\label{eq:ineq-a1b1ua2b2-6}
\E_{U_\epsilon A_1B_1}[\|p_{A_2B_2|U_\epsilon A_1B_1} - p_{A_2B_2}\|_1^3]\leq O(\epsilon^3).
\end{align}
Here~\eqref{eq:ineq-a1b1ua2b2-5} is a consequence of the monotonicity of Schatten one-norm under stochastic maps, and that 
$\|p_{A_1A_2B_1B_2}\cdot f_{A_1A_2B_1B_2}\|_1=O(1)$ which is derived from $\E[f^2]=1$.
To prove~\eqref{eq:ineq-a1b1ua2b2-6},
using $p_{A_1A_2B_1B_2}=p_{A_1B_1}\cdot p_{A_2B_2}$ for every $U_\epsilon=u\in \{\pm 1\}$  we have
$$p_{A_1A_2B_1B_2|u}=p_{A_1B_1}\cdot p_{A_2B_2}(1+\epsilon u f).$$
Therefore, for every $(A_1, B_1)=(a_1, b_1)$ we have 
$$p_{A_2B_2|a_1b_1u} = \frac{p_{A_2B_2}(1+\epsilon u f_{A_2B_2a_1b_1})}{1+ \epsilon u \,\E_{A_2B_2}[f_{A_2B_2a_1b_1}]} = p_{A_2B_2} + O(\epsilon),$$
where $f_{A_2B_2a_1b_1}$ is a function of $A_2B_2$ and is defined by restriction of $f$ to $(A_1, B_1)=(a_1, b_1)$. This gives~\eqref{eq:ineq-a1b1ua2b2-6}.
As a result, $\fS(A_1, B_1)\cap \fS(A_2, B_2)\subseteq \fS(A_1A_2, B_1B_2)$, which completes the proof.

%*************************rho-MC~ribbon**********************
\subsection{Maximal correlation and MC~ribbon}\label{subsec:rho-MC-ribbon}
As discussed in the introduction, maximal correlation is an important measure of correlation that similar to the MC~ribbon has the tensorization property and satisfies data processing inequality. Here we prove a connection between maximal correlation the our notation of MC~ribbon. 

The maximal correlation of a bipartite distribution $p_{AB}$ is defined in equation~\eqref{eq:max-correlatoin-31}. A simple algebra verifies that it is equivalently defined by 
\begin{align}\label{eq:max-correlatoin}
\rho(A, B):=& \max \quad \E[f_Ag_B],
\\&\text{subject to: } \E[f_A]=\E[g_B]=0, \nonumber\\
& \qquad \qquad \quad \E[f_A^2]=\E[g_B^2]=1,\nonumber
\end{align}
where maximum is taken over functions $f_A:\mathcal A\rightarrow \mathbb R$ and $g_B:\mathcal B\rightarrow \mathbb R$. 

From the Cauchy-Schwarz inequality we have $0\leq \rho(A, B)\leq 1$. Moreover, $\rho(A, B)=0$ if and only if $p_{AB}=p_A\cdot p_B$, and $\rho(A, B)=1$ if and only if $A$ and $B$ have a \emph{common data}~\cite{Witsenhausen}.
Moreover, maximal correlation is equal to the second singular value of a certain matrix in terms of distribution $p_{AB}$ and can be computed efficiently~(see e.g.~\cite{Kumar2, KangUlukus}).

It is known that maximal correlation can equivalently~\cite{Renyi1, Renyi2} be computed by 
\begin{align}\label{eq:rho2-fa-4}
\rho^2(A, B)  = &\max  \frac{\Var_{B}\E_{A|B}[f]}{\Var[f]}, 
\end{align}
where maximum is taken over all non-constant functions $f_A$. For the sake of completeness we give a proof of this fact in Appendix~\ref{app:max-correlation}.

In the above characterization of maximal correlation we observe similar terms as in the definition of the MC~ribbon. Indeed, if $f_A$ is a function of $A$ only, by the law of total variance we have $\Var[f]=\Var_A\E_{B|A}[f]$. Then if $(\lambda_1, \lambda_2)\in \fS(A, B)$, for such an $f_A$ we have
$$(1-\lambda_1)\Var[f]\geq \lambda_2\Var_B\E_{A|B}[f],$$
which using~\eqref{eq:rho2-fa-4} implies that $(1-\lambda_1)/\lambda_2\geq \rho^2(A, B)$ if $\lambda_2\neq 0$ and $\Var[f]\neq 0$. The following theorem states that this inequality in the other direction holds too. We leave the proof of this theorem for Appendix~\ref{app:max-correlation}.

\begin{theorem}\label{thm:rho-MC-ribbon}
For any bipartite distribution $p_{AB}$ we have 
$$ \rho^2(A, B)=\inf\frac{1-\lambda_1}{\lambda_2},$$
where infimum is taken over all $(\lambda_1, \lambda_2)\in \fS(A, B)$ with $\lambda_2\neq 0$.
\end{theorem}

The above theorem shows that maximal correlation can be characterized in terms of the MC~ribbon. 
Since by Theorem~\ref{thm:MC-ribbon-tensor-data} the MC~ribbon has the tensorization property and satisfies the data processing inequality, so does maximal correlation.

\begin{corollary}[\cite{Witsenhausen}] 
Maximal correlation has the following properties:
\begin{enumerate}
\item[\rm{(i)}]\emph{[Tensorization]} If $p_{A_1A_2B_1B_2}=p_{A_1B_1}\cdot p_{A_2B_2}$, then 
$$\rho(A_1A_2, B_1B_2)=\max\{\rho(A_1, B_1), \rho(A_2, B_2)\}.$$

\item[\rm{(ii)}]\emph{[Data processing]} If $p_{A_1A_2B_1B_2}=p_{A_1B_1}\cdot p_{A_2|A_1}\cdot p_{B_2|B_1}$, then 
$$ \rho(A_1, B_1)\geq\rho(A_2, B_2).$$
\end{enumerate}
\label{corol:rho-data-tensor-3}
\end{corollary}

Here is another consequence of the above theorem and Theorem~\ref{thm:MC-HC-ribbon} (relating the HC and MC ribbons). 

\begin{corollary}\label{corol:MC-subset-HC}
For any bipartite distribution $p_{AB}$ we have
$$\fR(A, B)\subseteq \fS(A, B),$$
and 
$$s^*(A,B):=\inf_{(\lambda_1,\lambda_2)\in \fR(A,B)} \frac{1-\lambda_1}{\lambda_2}\geq \,\rho^2(A,B).$$
\end{corollary}

We finish this section by pointing out that the MC~ribbon and HC~ribbon are not equal in general. Indeed there are examples~\cite[Section II A]{Anantharametal} of bipartite distributions $p_{AB}$ for which $s^*(A, B)$ is strictly greater than $\rho^2(A, B)$, which by Theorem~\ref{thm:rho-MC-ribbon} (that expresses maximal correlation in terms of the MC ribbon) gives $\fR(A, B)\neq \fS(A,B)$.

%*******************************MC~ribbon under wiring************************************
\section{MC~ribbon for no-signaling boxes}\label{sec:MC-ribbon-wiring}

In the same way we defined the HC~ribbon for no-signaling boxes, we may define the MC~ribbon for them too.

\begin{definition}
Given a no-signaling box $p(ab|xy)$, we define its MC~ribbon to be the intersection of the MC~ribbons of its outputs conditioned on all possible inputs, i.e.,
$$\mathfrak S(A, B|X, Y):=\bigcap_{x,y}\mathfrak S(A, B|X=x, Y=y).$$
We also define the maximal correlation of $p(ab|xy)$ to be the maximum of the maximal correlation of its outputs conditioned on all possible inputs, i.e.,
$$\rho(A,B|X, Y):= \max_{x, y} \rho(A, B|X=x, Y=y).$$
\end{definition}

Let us first state a variant of Theorem~\ref{thm:rho-MC-ribbon} (that expresses maximal correlation in terms of the MC ribbon)  for no-signaling boxes. Its proof is essentially the same as the proof of Theorem~\ref{thm:rho-MC-ribbon} and is presented in Appendix~\ref{app:inf-rho-mc-ribbon-NS}.

\begin{theorem}\label{thm:inf-rho-mc-ribbon-NS}
For any no-signaling box $p(ab|xy)$ we have
\begin{align*}
\inf \frac{1-\lambda_1}{\lambda_2} = \rho^2(A, B|XY),
\end{align*} 
where the infimum is taken over $(\lambda_1, \lambda_2)\in \fS(A, B|X, Y)$ with $\lambda_2\neq 0$.
\end{theorem}

Now we can prove a similar statement to Theorem~\ref{thm:main-theorem} (HC ribbon under wiring of no-signaling boxes) for the MC~ribbon of no-signaling boxes.

\begin{theorem}\label{thm:main-theorem-MC-ribbon}
Suppose that a no-signaling box $p(a'b'| x'y')$ can be generated from $n$ no-signaling boxes $p_i(a_ib_i| x_i y_i)$ where $i\in [n]$, under wirings. Then we have
\begin{align}\label{eq:mon-wiring-ribbon-MC}
\bigcap_{i=1}^n \fS(A_i,B_i|X_i, Y_i) \subseteq \fS(A',B'|X',Y').
\end{align}
\end{theorem}

\begin{proof} 
The proof of this theorem is similar to the proof of Theorem~\ref{thm:main-theorem}; we only need to replace mutual information with variance. Our main tool in the proof of Theorem~\ref{thm:main-theorem} is the chain rule, which here should be replaced by the law of total variance.
 
Here we provide an alternative proof by using Theorem~\ref{thm:MC-HC-ribbon} (that relates the MC ribbon to the HC ribbon), and adapting the proof of Theorem~\ref{thm:main-theorem}. Following the same steps as in the proof of Theorem~\ref{thm:main-theorem}, we need to show that 
$$(\lambda_1, \lambda_2)\in \fS(A_{[n]}X_{[n]}\Pi_{[n]}, B_{[n]}Y_{[n]} \Omega_{[n]}).$$

Let $f$ be a function of $A_{[n]}X_{[n]}\Pi_{[n]}B_{[n]}Y_{[n]} \Omega_{[n]}$ with $\E[f]=0$ and $\Var[f]=\E[f^2]=1$. We then define $p_{U_\epsilon A_{[n]}X_{[n]}\Pi_{[n]}, B_{[n]}Y_{[n]} \Omega_{[n]}}$ as in part (i) of the statement of Theorem~\ref{thm:MC-HC-ribbon}. Following  the proof of Theorem~\ref{thm:main-theorem},
we need to add the extra term
\begin{align}&
O\big(\E_{U_\epsilon S_iT_i\Pi_i\Omega_iX_iY_i}\big[\|p_{A_iB_i|U_\epsilon S_iT_i\Pi_i\Omega_iX_iY_i} - p_{A_iB_i|X_iY_i}\|_1^3\big]\big)\leq O(\epsilon^3),\label{eq:ineq-u-epsilon-10}
\end{align}
to the right hand side of~\eqref{eq:mc-12-lambda}; The inequality in~\eqref{eq:ineq-u-epsilon-10} is proved below. Then adding up the above inequalities  for $i=1,\dots, n$, we obtain
\begin{align*}&
\sum_{i=1}^n O\big(\E_{U_\epsilon S_iT_i\Pi_i\Omega_iX_iY_i}\big[\|p_{A_iB_i|U_\epsilon S_iT_i\Pi_i\Omega_iX_iY_i} 
- p_{A_iB_i|X_iY_i}\|_1^3\big]\big)
\leq O(n\epsilon^3)=O(\epsilon^3).
\end{align*}
Here we use the fact that $n$ although arbitrarily large, is a constant. 
The rest of the proof is identical to the proof of Theorem~\ref{thm:main-theorem}. 

It only remains to verify ~\eqref{eq:ineq-u-epsilon-10}.  Let us define
\begin{align*}&g_{A_iB_iX_iY_iS_iT_i\Pi_i\Omega_i}=\E_{A_{[n]}B_{[n]}X_{[n]}Y_{[n]}\Pi_{n]}\Omega_{[n]} |A_iB_iX_iY_iS_iT_i\Pi_i\Omega_i}[f].\end{align*}
Note that for every $U_\epsilon=u\in \{\pm 1\}$ we have
$$p_{A_{[n]}X_{[n]}\Pi_{[n]}, B_{[n]}Y_{[n]} \Omega_{[n]}|u} = p_{A_{[n]}X_{[n]}\Pi_{[n]}, B_{[n]}Y_{[n]} \Omega_{[n]}}(1+\epsilon u f).$$
Thus we can compute
\begin{align*}
p_{A_iB_iX_iY_iS_iT_i\Pi_i\Omega_i|u}&=p_{A_iB_iX_iY_iS_iT_i\Pi_i\Omega_i}(1+\epsilon u g)\\
&= p_{X_iY_iS_iT_i\Pi_i\Omega_i}\cdot p_{A_iB_i|X_iY_i}(1+\epsilon u g),
\end{align*}
where in the second line we use part (i) of Lemma~\ref{lem:dist-wiring}. As a result, for every $(X_i, Y_i, S_i, T_i, \Pi_i, \Omega_i)=(x_i, y_i, s_i, t_i, \pi_i, \omega_i)$ we have
\begin{align*}p_{A_iB_i|x_i y_i s_i t_i \pi_i \omega_i u}&=\frac{p_{A_iB_i|x_iy_i}(1+\epsilon u g_{A_iB_i x_i y_i s_i t_i \pi_i \omega_i})}{1+\epsilon u \E_{A_iB_i|x_i y_i}[g_{A_iB_i x_i y_i s_i t_i \pi_i \omega_i}]}
=p_{A_iB_i|x_iy_i}+ O(\epsilon).\end{align*}
Here for the second equality we use $g(a_ib_ix_i y_i s_i t_i \pi_i \omega_i)=O(1)$ which is implied by $\E[f^2]=1$. Equation~\eqref{eq:ineq-u-epsilon-10} is an immediate consequence of the above equation.

\end{proof}

The following corollary is a consequence of the above theorem and Theorem~\ref{thm:inf-rho-mc-ribbon-NS} (that expresses maximal correlation of a no-signaling box in terms of its MC ribbon).

\begin{corollary}\label{corol:rho-wiring-monotone}
Suppose that a no-signaling box $p(a'b'| x'y')$ can be generated from $n$ no-signaling boxes $p_i(a_ib_i| x_i y_i)$ where $i\in [n]$, under wirings. Then we have
\begin{align}
\rho(A', B'|X', Y')\leq \max_i \rho(A_i, B_i| X_i, Y_i).
\end{align}

\end{corollary}

We can now state the following corollary that is similar to Corollary~\ref{corol:HC-ribbon-close-wiring} but for maximal correlation.

\begin{corollary}\label{corol:rho-close-set}
Let $r\in [0,1]$ be arbitrary. Then the set of no-signaling boxes $p(a, b|x, y)$ with 
$\rho(A, B|X, Y)\geq r$ is closed under wirings.
\end{corollary}

The advantage of this corollary comparing to Corollary~\ref{corol:HC-ribbon-close-wiring} is that computing the maximal correlation of no-signaling boxes is a much easier problem than computing their HC~ribbon.

%*****************************Isotropic boxes*************************
\section{Example: isotropic boxes} \label{sec:isotropic}

Isotropic boxes are defined by
\begin{align}\label{eq:iso-box-2}
\PR_{\eta}(a,b|x,y) := \begin{cases}
\frac{1+\eta}{4} \qquad\qquad \text{if } a \oplus b= xy,\\
\frac{1-\eta}{4} \qquad\qquad \text{otherwise}.\\
\end{cases}
\end{align}
Here $a, b, x,y\in \{0,1\}$ and $0\leq \eta\leq 1$ is arbitrary.  Note that $\PR_\eta(a|x,y)$ and $\PR_{\eta}(b|x,y)$ are both the uniform
distribution independent of $x, y$, so $\PR_\eta$ is a no-signaling box. 
Here as an application of Corollary~\ref{corol:rho-wiring-monotone} we prove Conjecture~\ref{conj1} in the range $\eta_1, \eta_2\in [1/\sqrt 2, 1]$. We start with the case where the parties are not provided with common randomness. 

\begin{theorem}\label{thm:isotropic-box}
For $0\leq \eta_1<\eta_2\leq 1$, using an arbitrary number of copies of $\PR_{\eta_1}$, a single copy of $\PR_{\eta_2}$ cannot be generated under wirings. 
\end{theorem}

\begin{proof} 
Let $q_{\eta}(c, c')$ be the following distribution.
\begin{align*}
q_{\eta}(c, c') := \begin{cases}
\frac{1+\eta}{4} \qquad\qquad \text{if } c = c',\\
\frac{1-\eta}{4} \qquad\qquad \text{otherwise}.\\
\end{cases}
\end{align*} 
If $xy=0$ then the conditional distribution $\PR_{\eta}(a,b| x,y)$ is equal to $q_{\eta}(a,b)$, and if $xy=1$ it
coincides with $q_{\eta}(a\oplus 1, b)$. On the other hand a simple computation verifies that 
$\rho(q_\eta)=\eta$. As a result we have
$$\rho(\PR_{\eta_1}) = \eta_1< \eta_2=\rho(\PR_{\eta_2}).$$
The result then follows from Corollary~\ref{corol:rho-wiring-monotone}.
\end{proof}

We now handle the case where common randomness is also available. For this we use the notion of CHSH value of boxes with binary inputs and outputs defined by
$$\CHSH:= \frac{1}{4}\sum_{a, b, x, y} \delta_{a\oplus b, xy} p(a, b|x,y),$$
where $\delta_{a\oplus b, xy}=1$ if $a\oplus b=xy$, and $\delta_{a\oplus b, xy}=0$ otherwise. Before stating our result, we need the following lemma which  shows that among all no-signaling boxes of the same CHSH value, PR boxes have the smallest maximal correlation.

\begin{lemma}\label{lem:ribbon-rho-CHSH} Let $q(\cdot | \cdot)$ be an arbitrary no-signaling box with binary inputs and outputs. Suppose that $\CHSH(q)\geq (1+\eta)/2$ for some $1/\sqrt{2}\leq \eta\leq 1$. Then we have $\rho(q)\geq \eta=\rho(\PR_\eta)$.
\end{lemma}

The proof of the above lemma can be found in Appenix~\ref{app:rho-chsh}. Our result is as follows:

\begin{theorem}\label{thm:isotropic-box-com-rand}
For $1/\sqrt 2\leq \eta_1<\eta_2\leq 1$, using common randomness and an arbitrary number of copies of $\PR_{\eta_1}$, a single copy of $\PR_{\eta_2}$ cannot be generated under wirings. 
\end{theorem}

\begin{proof} Suppose that $\PR_{\eta_2}$ can be generated with common randomness and with some copies of $\PR_{\eta_1}$ under wirings.  
Let the common randomness shared between the two parties be $R$. Then for each $R=r$, the parties generate some box $q_r(\cdot|\cdot)$ such that 
$$\Pr_{\eta_2}(a, b| x, y) = \sum_r p(R=r) q_r(a, b| x, y).$$
Note that the CHSH value is a linear function. Moreover, we have $\CHSH(\PR_{\eta_2})=(1+\eta_2)/2$. Therefore,
$$\sum_r p(R=r) \CHSH(q_r) = (1+\eta_2)/2.$$
Thus for at least one value of $r$ with $p(R=r)\neq 0$ we have $\CHSH(q_r)\geq (1+\eta_2)/2$. Lemma \ref{lem:ribbon-rho-CHSH} shows that among all no-signaling boxes of the same CHSH value, PR boxes have the smallest maximal correlation. Thus, $\CHSH(q_r)\geq (1+\eta_2)/2$ implies that
\begin{align}\label{eq:frpr-eta}
\rho(q_r) \geq \rho(\Pr_{\eta_2})=\eta_2.
\end{align}
On the other hand, by assumption two parties having access to some copies of $\PR_{\eta_1}$ can generate $q_r$ (without common randomness). Then by Corollary~\ref{corol:rho-wiring-monotone} we have
$$\rho(\eta_1)=\eta_1\geq \rho(q_r).$$
This is in contradiction with~\eqref{eq:frpr-eta} since $\eta_2>\eta_1$.
\end{proof}

%*****************************Conclusion*************************
\section{Conclusion}\label{sec:conclusion}

In this paper we defined the notion of HC~ribbon for no-signaling boxes, and proved a data processing type monotonicity property for the HC~ribbon of no-signaling boxes under wirings. 

We also defined the notion of MC~ribbon for bipartite distributions and showed that it has the tensorization property and is monotone under local operations. Generalizing its definition for no-signaling boxes, we showed that similar to the HC~ribbon, MC~ribbon is also monotone under wirings of no-signaling boxes. As a consequence of this result, we proved that maximal correlation is also monotone under wirings.  

As an application of these results, we proposed a systematic method to construct sets of no-signaling boxes that are closed under wirings. Moreover, we proved a conjecture about simulating isotropic boxes with each other for certain range of parameters. This provides us with a continuum of sets of non-local boxes that are closed under wirings. The existence of such sets was also conjectured in~\cite{Review14}.

In the problem of simulating isotropic boxes with each other, we used maximal correlation together with HC~ribbon in order to prove Conjecture~\ref{conj1} (in the range of parameters $\eta_1, \eta_2\in [1/\sqrt 2, 1]$) when common randomness is available too. Interestingly the range of parameters for which we could prove this conjecture starts with $1/\sqrt 2$ which is the border point of quantum correlations. 

To prove Conjecture~\ref{conj1} for other values of $\eta_1, \eta_2\in [1/2, 1/\sqrt 2]$,  one approach is to, instead of maximal correlation, use 
the monotonicity of either the MC~ribbon or HC~ribbon under wirings. Another strategy is to use the parameter $s^*$ defined in Corollary~\ref{corol:MC-subset-HC}. We leave the study of this conjecture for the range of parameters $\eta_1, \eta_2\in [1/2, 1/\sqrt 2]$, and investigating the success of the above approaches for future works.

In general HC~ribbon cannot be used alone to study wirings of no-signaling boxes when common randomness is available. We leave a more systematic study of common randomness in wirings for future works too.

In this paper we studied wirings of bipartite boxes only. We however may consider wirings of multipartite boxes. We have defined the HC~ribbon and MC~ribbon in the multipartite case too~\cite{BGfuture} (multipartite HC ribbon is also defined independently in \cite{SKVA}). With this definition, the work of~\cite{Nair} extends to the multipartite case as well. We however do not know whether Theorem~\ref{thm:main-theorem} and Theorem~\ref{thm:main-theorem-MC-ribbon} (on the HC and MC ribbons under wiring of no-signaling boxes) can be extended to the multipartite case or not.

\vspace{.2in}
\noindent\textbf{Acknowledgements.} 
The authors are thankful to T. V\'ertesi, M. Navascu\'es, O. Etesami and O. Ahmadi for their comments on the early drafts of this paper, and to the anonymous reviewers for their helpful comments.

%%%%%%%%%%%%%%%%%*****Appendix*****%%%%%%%%%%%%%%%%%%
\appendix

\section{Maximal correlation under wiring of two boxes}\label{mc-wiring2}
Here we present a proof of equation~\eqref{eq:AX-BY-maximal-correlation}. 
The first inequality $\rho(A, B)\leq \rho(AX, BY)$ is a consequence of the known fact that maximal correlation is monotone under local stochastic maps (see Corollary~\ref{corol:rho-data-tensor-3}).
Then we need to verify the second inequality, that is summarized in the following lemma.

\begin{lemma}\label{lem:simple-wiring-rho}
Suppose that $q(abxy)=q(xy)p(ab|xy)$ such that $p(a|xy)=p(a|x)$ and $p(b|xy)=p(b|y)$. Then we have
$$\rho(AX, BY)\leq \max\{\rho(X, Y), \rho(A, B|X, Y)\}.$$
\end{lemma}

Our proof of this lemma borrows ideas from the proof of the tensorization property of maximal correlation given in~\cite{Kumar}.

\begin{proof}
Starting from the definition of maximal correlation, a simple algebra verifies that $\rho(A, B)$ is the smallest number $\rho$ such that for every $f_A, g_B$ we have
\begin{align}\label{eq:rho-inequality-def}
\E[f_Ag_B]\leq \E[f_A]\cdot \E[g_B] + \rho\Var[f_A]^{1/2}\Var[g_B]^{1/2}.
\end{align}
As a result, we need to show that for functions $f_{AX}$ and $g_{BY}$ we have
$$\E[fg]\leq \E_{AX}[f]\E_{BY}[g]+\rho\sqrt{\Var_{AX}[f]\Var_{BY}[g]},$$
where $\rho=\max\{\rho(X, Y), \rho(A, B|X, Y)\}$.
For this we compute
\begin{align*}
\E[fg]&= \E_{XY} \E_{AB|XY}[fg]\\
&\stackrel{\text{\rm (i)}}{\leq} \E_{XY}\left[ \E_{A|XY}[f]\cdot \E_{B|XY}[g] +\rho\sqrt{\Var_{A|XY}[f]\cdot \Var_{B|XY}[g]}   \,    \right]\\
&= \E_{XY}\left[ \E_{A|X}[f]\cdot\E_{B|Y}[g]\Big] 
+\rho\E_{XY}\Big[\sqrt{\Var_{A|X}[f]\cdot\Var_{B|Y}[g]}\,\right]       \\
&\stackrel{\text{\rm (ii)}}{\leq}  \E_{X}\E_{A|X}[f]\cdot\E_Y\E_{B|Y}[g] + \rho\sqrt{\Var_X\E_{A|X}[f]\cdot\Var_Y\E_{B|Y}[g]} +\rho\E_{XY}\Big[\sqrt{\Var_{A|X}[f]\cdot\Var_{B|Y}[g]}\,\Big]  \\
&\stackrel{\text{\rm (iii)}}{\leq} \E_{X}\E_{A|X}[f] \cdot \E_Y\E_{B|Y}[g] + \rho\sqrt{\Var_X\E_{A|X}[f]\cdot \Var_Y\E_{B|Y}[g]} +\rho\sqrt{\E_X\Var_{A|X}[f]\cdot \E_Y\Var_{B|Y}[g]}  \\
&\stackrel{\text{\rm (iv)}}{\leq} \E_{AX}[f]\cdot \E_{BY}[g] 
+\rho\cdot\Big(     
\sqrt{\Var_X\E_{A|X}[f]+\E_X\Var_{A|X}[f]}\sqrt{\Var_Y\E_{B|Y}[g]+ \E_Y\Var_{B|Y}[g]}\Big)                \\
&\stackrel{\text{\rm (v)}}{=}  \E_{AX}[f]\cdot \E_{BY}[g]+\rho\sqrt{\Var_{AX}[f]\Var_{BY}[g]}.
\end{align*}
Here in (i) we use~\eqref{eq:rho-inequality-def} for the conditional distribution $p_{AB|X=x,Y=y}$ for all $(x, y)$ and take average over all those inequalities. In (ii) we use~\eqref{eq:rho-inequality-def} for distribution $q_{XY}$ applied to functions $\E_{A|X}[f]$ and $\E_{B|Y}[g]$.
In (iii) and (iv) we use the Cauchy-Schwarz inequality, and in (v) we use the law of total variance.

\end{proof}

%*****************************Initial efficiency*************************
\section{Gray-Wyner problem and the HC~ribbon}\label{app:GW}

In this appendix we observe that the HC~ribbon is connected to the well-studied problem of Gray-Wyner which yields a tangible operational interpretation for the HC~ribbon. 

 The Gray-Wyner problem~\cite{GrayWyner} is a distributed source coding problem, consisting of a transmitter and two receivers. The transmitter has i.i.d.\ repetitions of two correlated sources $A_{[n]}, B_{[n]}$ and aims to send $A_{[n]}$ to the first receiver and $B_{[n]}$ to the second receiver. The transmitter can send a common message of rate $R_0$ over a noiseless channel to both the receivers, and private messages of rates $R_1$ and $R_2$ to the two receivers respectively. This is depicted in Figure~2. Then Gray and Wyner show that this is possible if and only if there exists some $p_{U|AB}$ such that
\begin{align}
R_0&\geq I(U;AB),\label{eq:gw-1}\\
R_1&\geq H(A|U),\\
R_2&\geq H(B|U).\label{eq:gw-3}
\end{align}
In particular, when $R_0=0$, we should have $R_1\geq H(A)$ and $R_2\geq H(B)$, which is consistent with Shannon's source compression theorem. Now suppose that we are allowing for some positive rate for common message $R_0>0$, and we are asking for the amount of reduction in private rates $R_1$ and $R_2$, i.e., how large $\Delta_{R_1}=H(A)-R_1$ and $\Delta_{R_2}=H(B)-R_2$ can be? These are two parameters that we would like to maximize simultaneously, but there is a tradeoff between them. Since by~\eqref{eq:gw-1}-\eqref{eq:gw-3} for some $p_{U|AB}$ we have
$$\Delta_{R_1}\leq I(U;A), \qquad \Delta_{R_2}\leq I(U;B), \qquad R_0\geq I(U;AB),$$
one can see that the non-negative triple $(\Delta_{R_1}, \Delta_{R_2}, R_0)$ is obtainable if and only if\footnote{This can be shown using the duality of linear programs.}
$$\lambda_1\Delta_{R_1}+\lambda_2\Delta_{R_2}\leq R_0, \qquad \forall (\lambda_1, \lambda_2)\in \fR(A, B).$$

As we use the HC~ribbon to study wirings of no-signaling boxes, we also notice that the Gray-Wyner problem is related to the principle of Information Casualty \cite{IC09} when the sources are independent. The authors have conjectured that this connection extends to correlated sources as well~\cite{OurICPaper}.

\begin{figure}
\begin{center}
\includegraphics[width=3.5in]{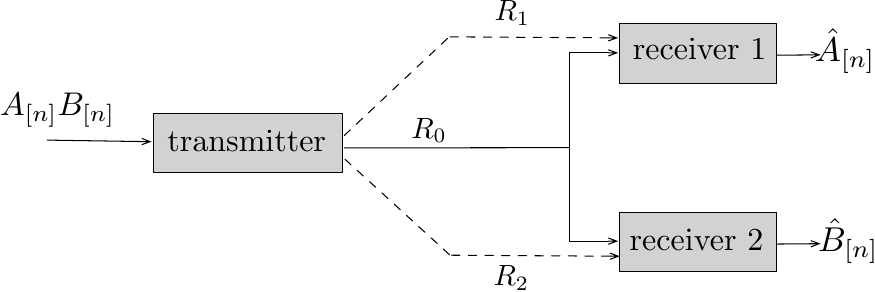}
\caption{The Gray-Wyner problem: A transmitter has two sources which should be sent to two receivers via three noiseless channels one of which is public and the other two are private. }
\end{center}
\end{figure}

%**************************Proof of Lemma************************************
\section{Proof of Lemma~\ref{lem:dist-wiring}}\label{app:lemm-proof}

Throughout this section, we need to write down the joint probability distribution of various random variables raised in wirings. First let us give some explanations on the validity of~\eqref{eq:p-a-x-m-n-y}. Suppose for a moment that all the boxes $p_i(a_ib_i|x_iy_i)=p_i(a_i|x_i)p_i(b_i|y_i)$ have the product form. So Alice and Bob are completely decoupled from each other. Then at her $j$-th action, Alice has $\wt_j$, and generates $\wpi_j$ and $\wx_j$. Then she puts $\wx_j$ as in the input of box $\wpi_j$ and observes $\wa_j$. Therefore, the joint distribution of random variables at Alice's side is
$$p(a_{[n]}x_{[n]}\pi_{[n]}|x') = \prod_{j=1}^n q(\wpi_j \wx_j|\wt_j x') p_{\wpi_j}(\wa_j|\wx_j).$$
We notice that $\wpi_1, \dots, \wpi_n$ is a permutation of $[n]$. So we may instead write this product with indices $j=\pi_i$. Then using $\wpi_{\pi_i}=i$ we obtain
\begin{align*}
p(a_{[n]}x_{[n]}\pi_{[n]}|x') & = \prod_{i=1}^n q(\wpi_{\pi_i} \wx_{\pi_i}|\wt_{\pi_i} x') p_{\wpi_{\pi_i}}(\wa_{\pi_i}|\wx_{\pi_i})\\
&= \prod_{i=1}^n q(i x_i|t_i x') p_{i}(a_{i}|x_{i}).
\end{align*}
By the same reasoning for random variables at Bob's side we have
\begin{align*}
p(b_{[n]}y_{[n]}\omega_{[n]}|y') &= \prod_{i=1}^n q(i y_i|s_i y') p_{i}(b_{i}|y_{i}).
\end{align*}
As a result when the boxes $p_i(a_ib_i|x_iy_i)$ have the product form $p_i(a_i|x_i)p_i(b_i|y_i)$, we have
\begin{align}
&p\big(a_{[n]}b_{[n]}x_{[n]}y_{[n]} \pi_{[n]}\omega_{[n]}| x'y'\big)= \prod_{i=1} p_i(a_ib_i|x_iy_i)q(i x_i|t_i x')q(i y_i|s_i y').\label{eq:24asdf}
\end{align}

Now consider general no-signaling boxes $p_i(a_ib_i|x_iy_i)$ that are not necessarily of product form. Note that to generate $a_{[n]}b_{[n]}x_{[n]}y_{[n]} \pi_{[n]}\omega_{[n]}$ Alice and Bob use each box once. Thinking of this box as a channel, the probability $p\big(a_{[n]}b_{[n]}x_{[n]}y_{[n]} \pi_{[n]}\omega_{[n]}| x'y'\big)$ should be linear in terms of box $p_i(a_ib_i|x_iy_i)$ for any $i$. On the other hand we showed that~\eqref{eq:24asdf} is valid for product boxes. Then the same equation holds if $p_i(a_ib_i|x_iy_i)$ is in the linear span of product boxes. On the other hand no-signaling boxes are in the linear span of product boxes.\footnote{Note that to write a general no-signaling box as a linear combination of product boxes we do need negative coefficients. This claim can be verified for example by computing the dimensions of linear span of product boxes and no-signaling boxes individually.} So~\eqref{eq:24asdf} holds for all no-signaling boxes.

Now, let us turn to the proof of  Lemma~\ref{lem:dist-wiring}. Since everything is conditioned on $x',y'$ for simplicity of notation we drop the conditionings on $x',y'$ and keep in mind that they are fixed. We also drop index $i$ in $p_i(\cdot |\cdot)$ as it is clear from other indices. We then have
\begin{align}
&p\big(a_{[n]}b_{[n]}x_{[n]}y_{[n]} \pi_{[n]}\omega_{[n]}\big)=\prod_{i=1}^n\bigg[ p\big(a_ib_i\,\big|\,x_iy_i\big) q\big(ix_i\,\big|\,t_i \big)q\big(iy_i\,\big|\,s_i\big)\bigg].\label{eq:p-a-x-m-n}
\end{align}

Recall that Alice uses the $i$-box in her action $\Pi_i=\pi_i$. Before using this box, Alice has used the boxes $\wpi_1, \dots, \wpi_{\pi_i-1}$, and  has the transcript $t_i=\wpi_{[\pi_i-1]}\wx_{[\pi_i-1]}\wa_{[\pi_i-1]}$. She then uses box $i$ with input $x_i$ with probability $q(ix_i|t_i)$. Similarly Bob uses box $i$ in step $\Omega_i=\omega_i$. Before using this box he uses boxes $\womega_1, \dots, \womega_{\omega_i-1}$, and has transcript $s_i=\womega_{[\omega_i-1]}\wy_{[\omega_i-1]}\wb_{[\omega_i-1]}$. Then he chooses box $i$ with input $y_i$ with probability $q(iy_i|s_i)$. When Alice and Bob both use the $i$-th box with inputs $x_i, y_i$, their outputs are $a_i, b_i$ with probability $p(a_ib_i| x_iy_i)$.

Imagine that Alice and Bob perform their wirings until they get to the $i$-th box, and then they stop.  Then by the above discussion and the no-signaling condition, the joint distribution of the outcomes of Alice and Bob is
\begin{align}
p(t_ia_ix_i \pi_i s_i y_ib_i\omega_i)&=\prod_{j\in M_i} p(a_jb_j|x_jy_j) q(jx_j|t_j)q(jy_j|s_j) \times\prod_{j\in \wpi_{[\pi_i]}\setminus M_i} p(a_j|x_j)q(jx_j|t_j)\nonumber\\&~\quad\times \prod_{j\in \womega_{[\omega_i]\setminus M_i}} p(b_j|y_j)p(j y_j|s_j),\label{eq:p-a-x-m-i}
\end{align}
where  $M_i:=  \wpi_{[\pi_i]}\cap \womega_{[\omega_i]}$. Note that $\wpi_{\pi_i}= \womega_{\omega_i}=i$, so $i\in M_i$.

Now imagine that Alice performs wirings until she gets to the $i$-th box and then she stops, but Bob uses all the boxes. Then again by the no-signaling condition the joint distribution of their outcomes is 
\begin{align}
&p\big(  t_{i} a_{i} x_i \pi_i  b_{[n]}y_{[n]}\omega_{[n]}   \big) =\prod_{j\in \wpi_{[\pi_i]}}p(a_j b_j|x_jy_j)q(jx_j|t_j)q(jy_j|s_j)\times \prod_{j\notin \wpi_{[\pi_i]}} p(b_j|y_j)q(jy_j|s_j).\label{eq:p-a-x-k-n-ell-i}
\end{align}

To obtain the next required marginal, imagine that Alice performs $i$ wirings (until the $i$-th action)  and then stops, i.e., Alice stops when she uses box $\wPi_i$. Assume that Bob uses all the boxes. Then Alice observes $\wt_i, \wa_i, \wx_i, \wpi_i$ and Bob observes $ b_{[n]}y_{[n]}\omega_{[n]}$, and the joint distribution of these outcomes is 
\begin{align}
&p\big(  \wt_{i} \wa_{i} \wx_i \wpi_i  b_{[n]}y_{[n]}\omega_{[n]}   \big) =\prod_{j\in \wpi_{[i]}}p(a_j b_j|x_jy_j)q(jx_j|t_j)q(jy_j|s_j)\times\prod_{j\notin \wpi_{[i]}} p(b_j|y_j)q(jy_j|s_j).\label{eq:p-a-x-k-n-ell-i-wpi}
\end{align}

We can now prove  Lemma~\ref{lem:dist-wiring}.

\vspace{.2in}
\noindent\emph{Proof of} (i): Observe that in~\eqref{eq:p-a-x-m-i} nothing is conditioned on $a_i, b_i$ as $i\in M_i$. Then we have
\begin{align}
p\big(t_ix_i \pi_is_iy_i\omega_i\big) & =\sum_{a_i, b_{i}}  p\big(t_ia_ix_i s_iy_ib_i\big)\nonumber\\
&=\prod_{j\in M_i\setminus\{i\}} p(a_jb_j|x_jy_j) q(jx_j|t_j)q(jy_j|s_j)  \times\prod_{j\in \wpi_{[\pi_i]}\setminus M_i} p(a_j|x_j)q(jx_j|t_j)\nonumber
\\&\quad~\times \prod_{j\in \womega_{[\omega_i]\setminus M_i}} p(b_j|y_j)p(j y_j|s_j).\label{eq:p-t-i-x-i}
\end{align}
Therefore,
\begin{align}\label{eq:abxsty89}
p\big(a_ib_{i}\big|\,t_ix_i\pi_i s_iy_i\omega_i\big) & =\frac{p\big(a_ib_{i}t_ix_i\pi_i s_iy_i\omega_i\big)}{p\big(t_ix_i\pi_i s_iy_i\omega_i\big) }
 = p\big(a_ib_i\big|\,x_iy_i\big).
\end{align}
As a result $H(A_iB_{i} | T_i\Pi_i S_{i} \Omega_iX_iY_i) = H(A_iB_{i} | X_iY_i)$,
 which gives the desired result.

\vspace{.2in}
\noindent\emph{Proof of} (ii): From~\eqref{eq:abxsty89} we find that
\begin{align*}
p\big(a_i\big|\,t_ix_i\pi_i s_iy_i\omega_i\big) & 
 = p\big(a_i\big|\,x_iy_i\big) = p(a_i\big|x_i).
\end{align*}
Therefore $H(A_i\big|\,T_iX_i\Pi_i S_iY_i\Omega_i\big)=H(A_i\big| X_i)$, or equivalently $I(A_i; T_i\Pi_i S_iY_i\Omega_i\big|X_i) = 0$
This gives $$I\big(A_i; S_i Y_i \Omega_i\big|\, T_i X_i \Pi_i\big)=0.$$ The other equality is proved similarly.

\vspace{.2in}
\noindent\emph{Proof of} (iii): Again using the fact that $i\in  \wpi_{[\pi_i]}$ and nothing in~\eqref{eq:p-a-x-k-n-ell-i} is conditioned on $a_i$ we have
\begin{align}
p\big(   t_{i}  x_i \pi_i  b_{[n]}y_{[n]}\omega_{[n]}   \big)&= \sum_{a_{i}} p\big(  t_{i} a_{i} x_i \pi_i  b_{[n]}y_{[n]}\omega_{[n]}   \big)\nonumber\\
& = p(b_i|x_iy_i)q(ix_i|t_j)q(iy_i|s_i) \times \prod_{j\in \wpi_{[\pi_i]}\setminus\{i\}}p(a_j b_j|x_jy_j)q(jx_j|t_j)q(jy_j|s_j)\nonumber\\ &~\quad \times \prod_{j\notin \wpi_{[\pi_i]}} p(b_j|y_j)q(jy_j|s_j).
\label{eq:ptxasetminusi}
\end{align}
We therefore have
\begin{align*}
p\big(  a_i| t_{i}  x_i \pi_i  b_{[n]}y_{[n]}\omega_{[n]}   \big) & =\frac{p\big(   t_{i} a_i x_i \pi_i  b_{[n]}y_{[n]}\omega_{[n]}   \big)}{p\big(   t_{i}  x_i \pi_i  b_{[n]}y_{[n]}\omega_{[n]}   \big)} \\
& = \frac{p\big(a_{i}b_i|x_iy_i\big)}{p\big(b_i|y_i\big)}\\
& = p(a_i| b_ix_iy_i).
\end{align*}
This means that  $$H(A_i|X_iB_iY_i) = H(A_i| T_iX_i\Pi_iB_{[n]}Y_{[n]}\Omega_{[n]}),$$ or equivalently
$$I\big(A_{i};  T_{i}\Pi_iB_{[n]}Y_{[n]}\Omega_{[n]}\big|\, X_iB_i Y_{i} \big)=0.$$
This gives $I\big(A_i;  B_{[n]}Y_{[n]}\Omega_{[n]}  \big|\, T_iX_i \Pi_i B_i Y_i\Omega_i\big)=0$. The other equality is proved similarly.

\vspace{.2in}
\noindent\emph{Proof of} (iv): In~\eqref{eq:p-a-x-k-n-ell-i-wpi} nothing is conditioned on $\tilde a_i$. Then we have
\begin{align*}
p\big(  \wt_{i} \wx_i \wpi_i  b_{[n]}y_{[n]}\omega_{[n]}   \big) = q(\wpi_i\wx_i|\wt_i)\times\prod_{j\in \wpi_{[i-1]}}p(a_j b_j|x_jy_j)q(jx_j|t_j)q(jy_j|s_j)\times \prod_{j\notin \wpi_{[i-1]}} p(b_j|y_j)q(jy_j|s_j),
\end{align*}
where we use $\wx_i=x_{\wpi_i}$ and $\wt_i=t_{\wt_i}$. Thus since in the above equation nothing is conditioned on $\wpi_i, \wx_i$ we obtain
$$p(\wpi_i\wx_i|  \wt_i  b_{[n]}y_{[n]}\omega_{[n]} ) = q(\wpi_i\wx_i|\wt_i).$$
This gives $H\big(\wPi_i\wX_i\big|\, \wT_{i}   B_{[n]}Y_{[n]}\Omega_{[n]} \big) = H\big(\wPi_i\wX_i|\wT_i\big)$ which is equivalent to 
$I\big(\wPi_i\wX_i; B_{[n]}Y_{[n]}\Omega_{[n]} \big|\, \wT_i\big)=0$. The other equation is proved similarly.

%**********************Proof of Lemma ******************************************
\section{Proof of Lemma~\ref{lemma:amnm}}\label{app:lemma-amnnproof}
First note that \eqref{eq:tyab5new} is implied by \eqref{eq:tyab5} by setting $U=A_{[n]}X_{[n]}\Pi_{[n]}$. Thus, we only need to prove \eqref{eq:tyab5}. By the chain rule we have
\begin{align}
I\big(U;A_{[n]}X_{[n]}\Pi_{[n]}|V\big)&=I\big(U;\wA_{[n]}\wX_{[n]}\wPi_{[n]}|V\big)\nonumber\\
&=\sum_{i=1}^n I\big(U;\wA_{i}\wX_{i}\wPi_{i}|\wA_{[i-1]}\wX_{[i-1]}\wPi_{[i-1]}V\big)
\nonumber\\&=\sum_{i=1}^n \bigg[I\big(U;\wX_{i}\wPi_i|\wA_{[i-1]}\wX_{[i-1]}\wPi_{[i-1]}V\big)+I\big(U;\wA_{i}|\wA_{[i-1]}\wX_{[i-1]}\wPi_{[i-1]} \wPi_{i}\wX_{i}V\big)\bigg]\nonumber\\
&=\sum_{i=1}^n \bigg[I\big(U;\wX_{i}\wPi_i|\wT_iV\big)+I\big(U;\wA_{i}|\wT_i \wPi_{i}\wX_{i}V\big)\bigg],\label{eq:cr345}
\end{align}
where in the last line we use $\wT_{i}=\wA_{[i-1]}\wX_{[i-1]}\wPi_{[i-1]}$.
Now note that 
\begin{align*}
\sum_{i=1}^nI\big(U;\wA_{i}|\wT_i \wPi_{i}\wX_{i}V\big) & = \sum_{i=1}^n \sum_{j=1}^n I\big(U;\wA_{i}|\wT_i \wX_{i}V, \wPi_{i}=j\big) p(\wPi_i=j) \\
 & = \sum_{i=1}^n \sum_{j=1}^n I\big(U;A_{j}|T_j X_{j}V, \wPi_{i}=j\big) p(\wPi_i=j) \\
& = \sum_{j=1}^n \sum_{i=1}^n  I\big(U;A_{j}|T_j X_{j}V, \Pi_{j}=i\big) p(\Pi_j=i) \\
& = \sum_{j=1}^n  I\big(U;A_{j}|T_j X_{j} \Pi_{j}V\big), 
\end{align*}
where in the second line we use $\wA_i=A_{\wPi_i}$, $\wX_i=X_{\wPi_i}$ and $\wT_i=T_{\wPi_i}$, and in the third line we use $\Pi_{\wPi_i}=i$. Since 
$T_i^{e}=T_i X_{i} \Pi_{i}$, we can write
$\sum_{i=1}^nI\big(U;\wA_{i}|\wT_i \wPi_{i}\wX_{i}V\big)= \sum_{i=1}^n  I\big(U;A_{i}|T^e_iV\big)$. Putting this in~\eqref{eq:cr345} we get equation~\eqref{eq:tyab5}.

%**********************Proof of the main theorem******************************************
\section{Proof of Theorem~\ref{thm:main-theorem}}\label{app:proof-main-thm}

To complete the proof of Theorem~\ref{thm:main-theorem} we need to show that $\chi(\lambda_1, \lambda_2)\geq 0$ for $\lambda_1, \lambda_2\in [0,1]$ with $\lambda_1+\lambda_2\geq 1$ where
\begin{align*}
\chi(\lambda_1, \lambda_2)&:=-\sum_{i=1}^n \bigg[\lambda_{1}I\big(U;\wX_{i}\wPi_i|\wT_i\big)+\lambda_2I\big(U;\wY_{i}\wOmega_i|\wS_i\big)\nonumber+\lambda_1I\big(U;A_{i};S^e_{i}\big|T_i^e\big) + \lambda_2 I\big(U;B_{i};T^e_{i}\big|\,S_i^e\big) \nonumber\\
&\qquad\qquad~ +I\big(U;A_iB_{i}\big|\,T^e_{i} S^e_i\big)\bigg]+I\big(U;A_{[n]}X_{[n]}\Pi_{[n]}B_{[n]}Y_{[n]}\Omega_{[n]}\big).
\end{align*}

$\chi(\lambda_1, \lambda_2)$ is an affine function of $(\lambda_1, \lambda_2)$. Moreover, extreme points of the convex set 
$$\{(\lambda_1, \lambda_2)\,|\, \lambda_1, \lambda_2\in[0,1]\, \&\, \lambda_1+\lambda_2\geq 1\},$$
are $(1, 0), (0,1)$ and $(1, 1)$. So it suffices to prove our claim for $(\lambda_1, \lambda_2)\in \{(1, 0), (0,1), (1, 1)\}$. The proof for $(\lambda_1, \lambda_2)=(0,1)$ is similar to that of $(\lambda_1, \lambda_2)=(1,0)$. So it suffices to prove $\chi(\lambda_1, \lambda_2)\geq 0$ when
$$\lambda_1=1, \quad \qquad \lambda_2\in\{0,1\}.$$

Using the chain rule we may write $\chi(1, \lambda_2)= \chi_A(1) + \chi_B(\lambda_2)$ where 
\begin{align*}
&\chi_A(1):=-\sum_{i=1}^n\bigg[ I\big(U;\wX_{i}\wPi_i|\wT_i\big)+I\big(U;A_{i};S^e_{i}\big|T_i^e\big) +I\big(U;A_i\big|\,T^e_{i} S^e_i\big) \bigg]+ I\big(U;A_{[n]}X_{[n]}\Pi_{[n]}\big),
\\
&
\\
&\chi_B(\lambda_2):=-\sum_{i=1}^n\bigg[ \lambda_2I\big(U;\wY_{i}\wOmega_i|\wS_i\big)+\lambda_2 I\big(U;B_{i};T^e_{i}\big|\,S_i^e\big) + I\big(U;B_{i}\big|\,T^e_{i} A_i S^e_i\big)\bigg]  
+I\big(U;B_{[n]}Y_{[n]}\Omega_{[n]} \big|\, A_{[n]}X_{[n]}\Pi_{[n]}\big).
\end{align*}
We start with $\chi_{A}(1)$. 
\begin{align*}
\chi_A(1)=&-\sum_{i=1}^n\bigg[ H(\wX_i\wPi_i\big|\, \wT_i) - H(\wX_i \wPi_i\big|\, \wT_i U) + I(A_i; S_i^e\big|\, T_i^e) - I(A_i ; S_i^e\big|\, T_i^eU)     +H(A_i\big|\, T_i^eS_i^e) - H(A_i\big|\, T_i^eS_i^e U)
    \bigg]
\\& ~+ H(A_{[n]}X_{[n]}\Pi_{[n]}) - H(A_{[n]}X_{[n]}\Pi_{[n]}\big|\,U). 
\end{align*}
Then we can write $\chi_A(1)=\phi_A(1) + \psi_A(1)$ where 
\begin{align*}
\phi_A(1):=& -\sum_{i=1}^n \bigg[H(\wX_i\wPi_i\big|\, \wT_i)  
+ I(A_i; S_i^e\big|\, T_i^e) +H(A_i\big|\, T_i^eS_i^e)  \bigg] +H(A_{[n]}X_{[n]}\Pi_{[n]}), \\
&\\
\psi_A(1):=& -\sum_{i=1}^n\bigg[  - H(\wX_i \wPi_i\big|\, \wT_i U)   - I(A_i ; S_i^e\big|\, T_i^eU)   - H(A_i\big|\, T_i^eS_i^e U)\bigg]- H(A_{[n]}X_{[n]}\Pi_{[n]}\big|\,U).   
\end{align*}

Using the second equation of Lemma \ref{lemma:amnm} with $V=\emptyset$, we get that
\begin{align}H(A_{[n]}X_{[n]}\Pi_{[n]})=\sum_{i=1}^n \bigg[H(\wX_i \wPi_i| \wT_i) + H(A_i| T_i^e)  \bigg].\label{eq:tyab5sym}\end{align}
Putting in $\phi_A(1)$ we find that 
\begin{align*}
\phi_A(1)= \sum_{i=1}^n \bigg[-H(\wX_i\wPi_i\big|\, \wT_i)  - I(A_i; S_i^e\big|\, T_i^e)
-H(A_i\big|\, T_i^eS_i^e) +H(\wX_i \wPi_i| \wT_i) + H(A_i| T_i^e) \bigg]=0.
\end{align*}
We can similarly use the second equation of Lemma \ref{lemma:amnm} with $V=U$ to write
\begin{align*}
\psi_A(1)&= \sum_{i=1}^n\bigg[   H(\wX_i \wPi_i\big|\, \wT_i U)   + I(A_i ; S_i^e\big|\, T_i^eU)   + H(A_i\big|\, T_i^eS_i^e U)
  - H(\wX_i \wPi_i\big|\, \wT_i U) -   H(A_i\big|\, T_i^e U)\bigg]    =0.  
\end{align*}
Therefore $\chi_A(1)=0$, and we need to show that $\chi(1, \lambda_2)=\chi_B(\lambda_2)\geq 0$, where we had
 \begin{align*}
\chi_B(\lambda_2)=&-\sum_{i=1}^n\bigg[ \lambda_2I\big(U;\wY_{i}\wOmega_i|\wS_i\big)+\lambda_2 I\big(U;B_{i};T^e_{i}\big|\,S_i^e\big) + I\big(U;B_{i}\big|\,T^e_{i} A_i S^e_i\big)\bigg]  \\
&~+I\big(U;B_{[n]}Y_{[n]}\Omega_{[n]} \big|\, A_{[n]}X_{[n]}\Pi_{[n]}\big).
\end{align*}
Again by the chain rule we may write $\chi_B(\lambda_2)=\phi_B(\lambda_2) + \psi_B(\lambda_2)$, where
\begin{align*}
\phi_B(\lambda_2):=&-\sum_{i=1}^n\bigg[ \lambda_2H\big(\wY_{i}\wOmega_i|\wS_i\big)+\lambda_2 I\big(B_{i};T^e_{i}\big|\,S_i^e\big) 
+ H\big(B_{i}\big|\,T^e_{i} A_i S^e_i\big)\bigg]  \\
&~+H\big(B_{[n]}Y_{[n]}\Omega_{[n]} \big|\, A_{[n]}X_{[n]}\Pi_{[n]}\big).
&\\
\psi_B(\lambda_2):=&-\sum_{i=1}^n\bigg[-\lambda_2H\big(\wY_{i}\wOmega_i|\wS_i U \big)-\lambda_2 I\big(B_{i};T^e_{i}\big|\,S_i^e U\big) 
- H\big(B_{i}\big|\,T^e_{i} A_i S^e_iU\big)\bigg]  \\
&~-H\big(B_{[n]}Y_{[n]}\Omega_{[n]} \big|\, A_{[n]}X_{[n]}\Pi_{[n]}U\big).
\end{align*}
Now note that by Lemma~\ref{lem:dist-wiring} part (ii) we have $I\big(B_i; T_i^e\big|\, S_i^e\big)=0$. Moreover, we can similarly use the second equation of Lemma \ref{lemma:amnm} with $V= A_{[n]}X_{[n]}\Pi_{[n]}$ to write
\begin{align*}H\big(B_{[n]}Y_{[n]}\Omega_{[n]} \big|\, A_{[n]}X_{[n]}\Pi_{[n]}\big) = \sum_{i=1}^n\bigg[  H\big(\wY_i\wOmega_i\big|\, A_{[n]}X_{[n]}\Pi_{[n]} \wS_i\big) + H\big(B_i\big|\,A_{[n]}X_{[n]}\Pi_{[n]} S_i^e \big)    \bigg].
\end{align*}
Therefore,
\begin{align}
\phi_B(\lambda_2)&=\sum_{i=1}^n\bigg[ -\lambda_2H\big(\wY_{i}\wOmega_i|\wS_i\big)- H\big(B_{i}\big|\,T^e_{i} A_i S^e_i\big)  
  + H\big(\wY_i\wOmega_i\big|\, A_{[n]}X_{[n]}\Pi_{[n]} \wS_i\big) 
+ H\big(B_i\big|\,A_{[n]}X_{[n]}\Pi_{[n]} S_i^e \big) \bigg]\nonumber\\& \geq
\sum_{i=1}^n\bigg[ -\lambda_2H\big(\wY_{i}\wOmega_i|\wS_i\big)- H\big(B_{i}\big|\,T^e_{i} A_i S^e_i\big)  
 + H\big(\wY_i\wOmega_i\big|\, A_{[n]}X_{[n]}\Pi_{[n]} \wS_i\big) 
 + H\big(B_i\big|\,A_{[n]}X_{[n]}\Pi_{[n]} T^e_{i} A_i S^e_i \big) \bigg]\nonumber\\
&= \sum_{i=1}^n (1-\lambda_2) H\big( \wY_i\wOmega_i\big|\, \wS_i    \big)\nonumber\\
&\geq  \sum_{i=1}^n (1-\lambda_2) H\big( \wY_i\wOmega_i\big|\, \wS_i   U \big),\label{eq:diadoa}
\end{align}
where in the  third  line we use Lemma~\ref{lem:dist-wiring} parts (iii) and (iv).
We continue
\begin{align}
\psi_B(&\lambda_2)\geq \sum_{i=1}^n\bigg[\lambda_2H\big(\wY_{i}\wOmega_i|\wS_i U \big) + H\big(B_{i}\big|\,T^e_{i} A_i S^e_iU\big)\bigg]   
 -H\big(B_{[n]}Y_{[n]}\Omega_{[n]} \big|\, A_{[n]}X_{[n]}\Pi_{[n]}U\big)\nonumber\\
&= \sum_{i=1}^n\bigg[\lambda_2H\big(\wY_{i}\wOmega_i|\wS_i U \big) + H\big(B_{i}\big|\,T^e_{i} A_i S^e_iU\big)  
 -H\big(\wY_i\wOmega_i\big|\, A_{[n]}X_{[n]}\Pi_{[n]} \wS_iU\big)- H\big(B_i \big|\, A_{[n]}X_{[n]}\Pi_{[n]}S_i^eU  \big)  \bigg] \label{eqn:useoflemma} \\
&\geq  \sum_{i=1}^n\bigg[\lambda_2H\big(\wY_{i}\wOmega_i|\wS_i U \big)  -H\big(\wY_i\wOmega_i\big|\, A_{[n]}X_{[n]}\Pi_{[n]} \wS_iU\big)   \bigg]\nonumber \\
&\geq  \sum_{i=1}^n(\lambda_2-1)H\big(\wY_{i}\wOmega_i|\wS_i U \big),\label{eq:diadoa222}
\end{align}
where in \eqref{eqn:useoflemma} we use the second equation of Lemma \ref{lemma:amnm} with $V= A_{[n]}X_{[n]}\Pi_{[n]}U$.
Comparing~\eqref{eq:diadoa222} and~\eqref{eq:diadoa} we conclude that $\chi(1, \lambda_2)=\chi_B(\lambda_2) = \phi_B(\lambda_2)+ \psi_B(\lambda_2)\geq 0$. We are done.

%****************************MC-tensor-data-direct-proof***********************************************
\section{Direct proof of Theorem~\ref{thm:MC-ribbon-tensor-data}}\label{app:MC-ribbon-tensor-data}

We will use the following lemma in the proof.

\begin{lemma}\label{lem:abc-e-var}
Suppose that $f_{ABC}$ is a function and $p_{ABC}=p_{AB}\cdot p_{C|A}$, i.e., $B$ and $C$ are independent conditioned on $A$. Then we have
$$\E_{AB}\Var_{C|AB}[f] \geq \E_A \Var_{C|A} \E_{B|AC}[f].$$
\end{lemma}

\begin{proof}
we compute
\begin{align*}
\E_{AB}\Var_{C|AB}[f] & = \E_{AB} \E_{C|AB}[(f-\E_{C|AB}[f])^2]\\
&=\E_{A}\E_{B|A} \E_{C|A}[(f-\E_{C|A}[f])^2]\\
& \geq \E_{A} \E_{C|A}[(\E_{B|A}f-\E_{BC|A}[f])^2] \\
& = \E_A \Var_{C|A} \E_{B|AC}[f],
\end{align*}
where in the third line we use the convexity of $t\mapsto t^2$.
\end{proof}

\noindent\emph{Proof of} (i): Let $(\lambda_1, \lambda_2)\in \mathfrak S(A_1,B_1)\cap\mathfrak S(A_2, B_2)$. Let $f_{A_1A_2, B_1B_2}$ be an arbitrary function. By the law of total variance we have
\begin{align}
\Var[f]&= \Var_{A_1B_1} \E_{A_2B_2|A_1B_1}[f] + \E_{A_1B_1}\Var_{A_2B_2|A_1B_1}[f]\nonumber\\
& \geq \lambda_1 \Var_{A_1} \E_{B_1|A_1}\E_{A_2B_2|A_1B_1}[f] + \lambda_2 \Var_{B_1}\E_{A_1|B_1}\E_{A_2B_2|A_1B_1}[f]\nonumber\\
&\quad~ +\E_{A_1B_1}\bigg(   \lambda_1 \Var_{A_2|A_1B_1}\E_{B_2|A_1A_2B_1}[f]    +\lambda_2 \Var_{B_2|A_1B_1}\E_{A_2|A_1B_1B_2}[f]    \bigg)\label{eqn:azxc1}\\
& \geq \lambda_1\bigg(   \Var_{A_1}\E_{A_2|A_1} \E_{B_1B_2|A_1A_2}[f] + \E_{A_1}\Var_{A_2|A_1} \E_{B_1|A_1} \E_{B_2|A_1A_2B_1}[f]     \bigg)\nonumber\\
& \quad+ \lambda_2\bigg(   \Var_{B_1}\E_{B_2|B_1} \E_{A_1A_2|B_1B_2}[f]  + \E_{B_1}\Var_{B_2|B_1} \E_{A_1|B_1} \E_{A_2|A_1B_1B_2}[f]     \bigg)\label{eqn:azxc2}\\
& = \lambda_1 \Var_{A_1A_2}\E_{B_1B_2|A_1A_2}[f] + \lambda_2\Var_{B_1B_2}\E_{A_1A_2|B_1B_2}[f].\label{eqn:azxc3}
\end{align}
Here \eqref{eqn:azxc1} follows from $(\lambda_1, \lambda_2)\in \mathfrak S(A_1,B_1)$ used for function $\E_{A_2B_2|A_1B_1}[f]$ on $\mathcal{A}_1\times\mathcal{B}_1$, and from $(\lambda_1, \lambda_2)\in\mathfrak S(A_2, B_2)$ used for function $f$ restricted on $\{(a_1, b_1)\}\times \mathcal{A}_2\times\mathcal{B}_2$ for any $(a_1, b_1)\in \mathcal A_1\times \mathcal B_1$. Note that in the latter case we use the fact that the conditional distribution $p_{A_2B_2|a_1b_1}$ is independent of $(a_1, b_1)$.
For~\eqref{eqn:azxc2} we use Lemma~\ref{lem:abc-e-var}. Finally,~\eqref{eqn:azxc3}  follows from the law of total variance.  We conclude that $\mathfrak S(A_1,B_1)\cap\mathfrak S(A_2, B_2)\subseteq \fS(A_1A_2, B_1B_2)$. For the inclusion in the other direction it suffices to consider those $f$ that are only a function of $A_iB_i$, $i=1, 2$.

\vspace{.15in}
\noindent\emph{Proof of} (ii):  Let $(\lambda_1, \lambda_2)\in \fS(A_1, B_1)$. Let $f_{A_2B_2}$ be an arbitrary function. Since $A_2$ and $B_2$ are independent conditioned on $A_1B_1$ and the MC~ribbon of independent random variables is the whole $[0,1]^2$, when we condition on 
$(A_1, B_1)=(a_1, b_1)$, the MC~ribbon of $(A_2, B_2)$ will include the pair $(\lambda_1, \lambda_2)$. Hence, for every $(A_1, B_1)=(a_1, b_1)$ we have 
\begin{align*}
\Var_{A_2B_2|a_1b_1}[f]\geq &\lambda_1  \Var_{A_2|a_1b_1}\E_{B_2|A_2, a_1b_1}[f] 
+ \lambda_2 \Var_{B_2|a_1b_1}\E_{A_2|B_2, a_1b_1}[f].
\end{align*}
Then by taking average over $A_1, B_1$ we have
\begin{align}
&\E_{A_1B_1}\Var_{A_2B_2|A_1B_1}[f]\geq \lambda_1 \E_{A_1B_1} \Var_{A_2|A_1B_1}\E_{B_2|A_1B_1A_2}[f] + \lambda_2 \E_{A_1B_1} \Var_{B_2|A_1B_1}\E_{A_2|A_1B_1B_2}[f].\label{eq:ineq-mc-r-1}
\end{align}

Define $\tilde f_{A_1B_1}:=\E_{A_2B_2|A_1B_1}[f]$. Then since $(\lambda_1, \lambda_2)\in \fS(A_1, B_1)$ we have
$$\Var[\tilde f] \geq \lambda_1 \Var_{A_1}\E_{B_1|A_1}[\tilde f] + \lambda_2 \Var_{B_1}\E_{A_1|B_1}[\tilde f],$$
which is equivalent to
\begin{align}
&\Var_{A_1B_1}\E_{A_2B_2|A_1B_1}[f] \geq \lambda_1 \Var_{A_1}\E_{B_1A_2B_2|A_1}[ f] + \lambda_2 \Var_{B_1}\E_{A_1A_2B_2|B_1}[ f].\label{eq:ineq-mc-r-2}
\end{align}
Summing~\eqref{eq:ineq-mc-r-1} and~\eqref{eq:ineq-mc-r-2} and using the law of total variance we obtain
\begin{align*}
\Var[f] & \geq \lambda_1\bigg(  \E_{A_1B_1} \Var_{A_2|A_1B_1}\E_{B_2|A_1B_1A_2}[f] 
+   \Var_{A_1}\E_{B_1A_2B_2|A_1}[ f]\bigg)\\
&\quad~+\lambda_2\bigg(  \E_{A_1B_1} \Var_{B_2|A_1B_1}\E_{A_2|A_1B_1B_2}[f]+   \Var_{B_1}\E_{A_1A_2B_2|B_1}[ f]     \bigg)\\
& \geq \lambda_1\bigg(  \E_{A_1} \Var_{A_2|A_1} \E_{B_1|A_1}\E_{B_2|A_1B_1A_2}[f] +   \Var_{A_1}\E_{A_2|A_1}\E_{B_1B_2|A_1A_2}[ f]\bigg)\\
& \quad~ +\lambda_2\bigg(  \E_{B_1} \Var_{B_2|B_1}\E_{A_1|B_1}\E_{A_2|A_1B_1B_2}[f]+   \Var_{B_1}\E_{B_2|B_1}\E_{A_1A_2|B_1B_2}[ f]     \bigg)\\
& = \lambda_1 \Var_{A_1A_2} \E_{B_1B_2|A_1A_2}[f] + \lambda_2 \Var_{B_1B_2}\E_{A_1A_2|B_1B_2}[f]\\
& \geq \lambda_1\Var_{A_2}\E_{A_1|A_2}\E_{B_1B_2|A_1A_2}[f] 
+\lambda_2 \Var_{B_2}\E_{B_1|B_2}\E_{A_1A_2|B_1B_2}[f]\\
& = \lambda_1\Var_{A_2}\E_{B_2|A_2}[f] +\lambda_2 \Var_{B_2}\E_{A_2|B_2}[f],
\end{align*}
where in the second line we use Lemma~\ref{lem:abc-e-var}, and in the last line we use the fact that $f$ is a function of $A_2, B_2$ only. Therefore, $(\lambda_1, \lambda_2)\in \fS(A_2, B_2)$.

%*************************************maximal correlation**************************************************
\section{Proof of Theorem~\ref{thm:rho-MC-ribbon}}\label{app:max-correlation}
Let us first give a proof of~\eqref{eq:rho2-fa-4}.
In definition~\eqref{eq:max-correlatoin} of maximal correlation we may drop one of the two constraints $\E[f_A]=0$ and $\E[g_B]=0$, i.e., for instance we can write
\begin{align}\label{eq:max-correlatoin-def3}
\rho(A, B)=& \max \quad \E[f_Ag_B],
\\&\text{subject to: } \E[f_A]=0, \nonumber\\
& \qquad \qquad \quad \E[f_A^2]=\E[g_B^2]=1.\nonumber
\end{align}
This is because if $\E[f_A]=0$, for an arbitrary $g_B$ if we let $\tilde g_B:= g_B-\E[g_B]$, then $\E[\tilde g_B]=0$ and $\E[f_Ag_B]= \E[f_A\tilde g_B]$ as well as 
$\E[\tilde g_B^2] =\Var[g_B]\leq \E[g_B^2]$; Then we can scale $\tilde g_B$ to make $\E[\tilde g_B^2]$ to be one, while increasing $\E[f_A\tilde g_B]$.

Let us fix $f_A$ and try to maximize $\E[f_Ag_B]$ over all $g_B$ with $\E[g_B^2]=1$. By the Cauchy-Schwarz inequality we have 
\begin{align*}
\E[f_Ag_B]& = \E_B\E_{A|B}[f_Ag_B] \\
&=  \E_B[\E_{A|B}[f_A]g_B] \\
&\leq \E_{B}[(\E_{A|B}[f_A])^2]^{1/2} \cdot \E_B[g_B^2]^{1/2}\\
& =\E_{B}[(\E_{A|B}[f_A])^2]^{1/2}. 
\end{align*}
Moreover, letting $g_B=\alpha \E_{A|B}[f_A]$, for the appropriate choice of constant $\alpha$, the above upper bound is attained. As a result we have
\begin{align}
\rho^2(A, B)  = &\max  \quad\E_{B}[\E_{A|B}[f_A]^2]\label{eq:def-max-corr-5}\\
& \text{subject to: } \E[f_A]=0,\nonumber\\
&  \qquad \qquad \quad  \E[f_A^2]=1.\nonumber
\end{align}
We may rewrite the above optimization in terms of variance to remove the constraints.
\begin{align}\label{eq:def-max-corr-6}
\rho^2(A, B)  = &\max  \frac{\Var_{B}\E_{A|B}[f]}{\Var[f]}, 
\end{align}
where maximum is taken over all non-constant functions $f_A$.

We now give the proof of Theorem~\ref{thm:rho-MC-ribbon}.\\

\noindent \emph{Proof of Theorem~\ref{thm:rho-MC-ribbon}.} 
Let $(\lambda_1, \lambda_2)\in \fS(A, B)$ where $\lambda_2\neq 0$. By definition we have 
$$\Var[f] \geq \lambda_1 \Var_{A}[\E_{B|A}[f]] + \lambda_2 \Var_{B}[\E_{A|B}[f]].$$
Assuming that $f=f_A$ is a function of $A$ only, we find that $\Var[f]=\Var_{A}\E_{B|A}[f]$. Therefore,
$$\frac{1-\lambda_1}{\lambda_2} \Var[f]\geq \Var_{B}\E_{A|B}[f].$$
Comparing to~\eqref{eq:def-max-corr-6} we find that 
$$(1-\lambda_1)/\lambda_2\geq \rho^2(A, B)=:\rho^2.$$

For the other direction, let $\epsilon>0$ be a constant, and let $n$ be some integer. Define 
$$\lambda_1^{(n)} = 1-\frac{\rho^2+\epsilon}{n},\qquad\lambda_2^{(n)}=\frac{1}{n}.$$
We claim that for sufficiently large $n$, $(\lambda_1^{(n)}, \lambda_2^{(n)})$ is in $\fS(A, B)$. Otherwise there is a function $f_{AB}$ such that  
\begin{align}\label{eq:g-j-ineq}
\Var[f] <\lambda_1^{(n)} \Var_{A}\E_{B|A}[f] + \lambda_2^{(n)} \Var_{B}\E_{A|B}[f].
\end{align}
Note that $\Var[f]\neq 0$ because otherwise the right hand side would have been zero too which is in contradiction with the strict inequality.
Thus with no loss of generality we may assume that 
$$\Var[f]=1.$$
Using the law of total variance
$$\Var[f]=\Var_A\E_{B|A}[f] + \E_A\Var_{B|A}[f],$$
equation~\eqref{eq:g-j-ineq} can equivalently be written as
\begin{align}\label{eq:g-j-ineq-2}
\frac{1-\lambda_1^{(n)}}{\lambda_2^{(n)}} \Var_{A}\E_{B|A}[f] + \frac{1}{\lambda_2^{(n)}}\E_{A}\Var_{B|A}[f]< \Var_{B}\E_{A|B}[f].
\end{align}
We have $\Var_{B}\E_{A|B}[f]\leq \Var[f]=1$, and $1/\lambda_2^{(n)}=n$. Therefore,
\begin{align}\label{eq:x2gj-3}
\E_{A}\Var_{B|A}[f] <1/n.
\end{align}
Let us define
$\tilde f=\E_{B|A}[f]$. Observe that $\tilde f$ is a function of $A$ only, $\E[\tilde f]=\E[f]$, and $ \Var_{A}\E_{B|A}[f] =\Var[\tilde f]$. Moreover,~\eqref{eq:x2gj-3} is equivalent to 
$$\E[(f-\tilde f)^2]<1/n.$$
Thus from~\eqref{eq:g-j-ineq-2} and using the fact that $(1-\lambda_1^{(n)})/\lambda_2^{(n)}=\rho^2+\epsilon$ we have
\begin{align*}
(\rho^2+\epsilon)\Var[\tilde f] & < \Var_{B}\E_{A|B}[f]\\
& =  \Var_{B}\E_{A|B}[(f-\tilde f)+ \tilde f]\\
& = \E_{B}\big[   \big(   \E_{A|B}[f-\tilde f] + \E_{A|B}[\tilde f]  -\E[\tilde f]     \big)^2   \big]\\
& = \E_{B}\bigg[   \big(   \E_{A|B}[f-\tilde f]\big)^2 + \big(\E_{A|B}[\tilde f]  -\E[\tilde f] \big)^2 +   2 \big(\E_{A|B}[f-\tilde f]\big)\big(\E_{A|B}[\tilde f]  -\E[\tilde f] \big) \bigg]\\
& \leq \E[(f-\tilde f)^2] + \Var_{B}\E_{A|B}[\tilde f] + 2 \bigg( \E_{B}\big[\big(\E_{A|B}[f-\tilde f]\big)^2\big]\cdot \E_{B}\big[\big(\E_{A|B}[\tilde f]  -\E[\tilde f] \big)^2\big]     \bigg)^{1/2}\\&
\leq \frac{1}{n} +\Var_{B}\E_{A|B}[\tilde f]  + 2 \bigg( \E_{B}\E_{A|B}\big[\big([f-\tilde f]\big)^2\big]\cdot \E_{B}\big[\big(\E_{A|B}[\tilde f]  -\E[\tilde f] \big)^2\big]     \bigg)^{1/2}\\
& \leq \frac{1}{n} +\Var_{B}\E_{A|B}[\tilde f]  + 2\bigg(\frac{1}{n} \Var_{B}\E_{A|B}[\tilde f] \bigg)^{1/2}\\
& \leq \Var_{B}\E_{A|B}[\tilde f]  + \frac{3}{\sqrt{n}},
\end{align*}
where in the last line we use 
$$\Var_{B}\E_{A|B}[\tilde f]\leq \Var[\tilde f] = \Var_{A}\E_{B|A}[f] \leq \Var[f]=1.$$
We also notice that $\tilde f$ is not constant because using~\eqref{eq:x2gj-3} we have
$$\Var[\tilde f] =  \Var_{A}\E_{B|A}[f] = \Var[f] - \E_{A}\Var_{B|A}[f]  >1-\frac{1}{n}.$$
Therefore, using~\eqref{eq:def-max-corr-6} we have
$$\rho^2+\epsilon \leq \frac{\Var_{B}\E_{A|B}[\tilde f]}{\Var[\tilde f]} + \frac{3}{\Var[\tilde f]\sqrt{n}}\leq \rho^2 + \frac{3}{(1-1/n)\sqrt{n}},$$
which does not hold for sufficiently large $n$. We conclude that for sufficiently large $n$ the point $(\lambda_1^{(n)}, \lambda_2^{(n)})$  belongs to $\fS(A, B)$. As a result, we have
$$\inf \frac{1-\lambda_1}{\lambda_2}\leq \rho^2+\epsilon,$$
for every $\epsilon>0$. Then 
$$\inf \frac{1-\lambda_1}{\lambda_2}\leq \rho^2.$$
We are done.
\hfill$\Box$

%******************************maximal correlation for NS boxes****************************************

\section{Proof of Theorem~\ref{thm:inf-rho-mc-ribbon-NS}}\label{app:inf-rho-mc-ribbon-NS}

As shown in the proof of Theorem~\ref{thm:rho-MC-ribbon} for every $x, y$ and $(\lambda_1, \lambda_2)\in \fS(A, B|X=x, Y=y)$ we have
$(1-\lambda_1)/\lambda_2\geq \rho^2(A, B|X=x, Y=y).$
Therefore, for  $(\lambda_1, \lambda_2)\in \fS(A, B|X,Y)$ we have
$$\frac{1-\lambda_1}{\lambda_2}\geq \rho^2(A, B|X=x, Y=y),$$
for every $x, y$. Taking the maximum of the right hand side over all $x, y$ we obtain
$$\inf \frac{1-\lambda_1}{\lambda_2} \geq \rho^2(A, B|XY),$$
where the infimum is taken over $(\lambda_1, \lambda_2)\in \fS(A, B|X, Y)$ with $\lambda_2\neq 0$.

For the other direction, recall that in the proof of Theorem~\ref{thm:rho-MC-ribbon} we show that $(\lambda_1^{(n)}, \lambda_2^{(n)})$  defined by
$$\lambda_1^{(n)}=1-\frac{\rho^2(A, B|X=x, Y=y) +\epsilon}{n} , \qquad \lambda_2^{(n)}=\frac{1}{n},$$
for a given $\epsilon>0$, belongs to $\fS(A, B|X=x, Y=y)$ for sufficiently large $n$. Since 
$\rho^2(A, B|X, Y)\geq \rho^2(A, B|X=x, Y=y)$ we find that $(\tilde \lambda_1^{(n)}, \tilde \lambda_2^{(n)})$ defined by
$$\tilde \lambda_1^{(n)}=1-\frac{\rho^2(A, B|X, Y) +\epsilon}{n} , \qquad \tilde \lambda_2^{(n)}=\frac{1}{n},$$
belongs to $\fS(A, B|X=x, Y=y)$ for sufficiently large $n$ too. As a result, $(\tilde \lambda_1^{(n)}, \tilde \lambda_2^{(n)})$ belongs to 
$\fS(A, B|X, Y)$ for sufficiently large $n$. We conclude that 
$$\inf \frac{1-\lambda_1}{\lambda_2} \leq \frac{1-\tilde \lambda_1^{(n)}}{\tilde \lambda_2^{(n)}}=\rho^2(A, B|XY)+\epsilon,$$
for every $\epsilon>0$. We are done.

%***************************MC ribbon- HC ribbon*********************************************
\section{Proof of Theorem~\ref{thm:MC-HC-ribbon}}\label{app:MC-HC-ribbon-proof}
Given (ii) the proof of (iii) is immediate; We only need to put $K=0$.
 
Let us denote the set of pairs $(\lambda_1, \lambda_2)$ described in parts (i) and (ii) of the theorem by $\fS_{\text{\rm i}}(A, B)$ and $\fS_{\text{\rm ii}}(A, B)$ respectively. 
We need to show $\fS(A,B) = \fS_{\text{\rm i}}(A, B)=\fS_{\text{\rm ii}}(A, B)$.

For a bipartite distribution $p_{AB}$ we let $\supp(p_{AB})\subseteq \mathcal A\times \mathcal B$ be the set of pairs $(a, b)$ such that $p(ab)\neq 0$. We also let $\mathcal W(p_{AB})$ be the set of distributions $q_{AB}$ with $\supp(q_{AB})=\supp(p_{AB})$. Note that
perturbations of the form~\eqref{sweepsupport} sweep the neighborhood of $p_{AB}$ in $\mathcal W(p_{AB})$.
In the following we also use the notation $q_{AB|U}\in \mathcal W(p_{AB})$ for some distribution $q_{ABU}$ by which we mean that for every $u\in \mathcal U$ the conditional distribution $q_{AB|U=u}$ is in $\mathcal W(p_{AB})$.  

Observe that $\Upsilon: \mathcal W(p_{AB})\rightarrow \mathbb R$ is a smooth function. So for $q_{AB}\in \mathcal W(p_{AB})$ letting $v_{AB}=q_{AB}-p_{AB}$ we may write 
$$\Upsilon(q_{AB}) = \Upsilon(p_{AB}) +  D^{(1)}_v(p_{AB}) + \frac{1}{2}D^{(2)}_v(p_{AB})  + O(\|v\|_1^3),$$
where $D^{(1)}_v(p_{AB})$ and $D^{(2)}_v(p_{AB})$ are respectively the first and second derivatives of $\Upsilon$ at $p_{AB}$ in the direction of $v$.  Observe that $D^{(1)}_v(p_{AB})$ and $D^{(2)}_v(p_{AB})$ are not infinity\footnote{The derivative of entropy function is infinity only when we change the distribution by making a non-zero probability equal to zero, or vice versa.} since $q_{AB}\in \mathcal W(p_{AB})$.  

 In the following we will show that $\fS_{\text{\rm i}}(A, B)\subseteq \fS(A, B)\subseteq \fS_{\text{\rm ii}}(A, B)\subseteq \fS_{\text{\rm i}}(A, B)$ which finishes the proof.  

\vspace{.15in}
\noindent
\emph{Proof of $\fS_{\text{\rm ii}}(A, B)\subseteq \fS_{\text{\rm i}}(A, B)$}:
From the definitions it is clear that $\fS_{\text{\rm ii}}(A, B)$ is more restrictive than $\fS_{\text{\rm i}}(A, B)$. We only need to note that for $f_{AB}$ with $\E[f]=0$ and $\Var[f]=\E[f^2]=1$, and for every $U_\epsilon=u$ we have
$$\|p_{AB|u}  - p_{AB} \|_1^3= \epsilon^3\|p_{AB}\cdot f_{AB}\|_1^3 = O(\epsilon^3),$$
where in the last step we use $f(ab)=O(1)$ for every $a, b$, which is implied by $\E[f^2]=1$.

\vspace{.15in}
\noindent
\emph{Proof of $\fS(A, B)\subseteq \fS_{\text{\rm ii}}(A, B)$}:
Observe that in the definition of $\fS_{\text{\rm ii}}(A, B)$, without loss of generality, we can restrict ourselves to $p_{AB|U}\in \mathcal W(p_{AB})$. In other words, if we have the inequality for $p_{AB|U}\in \mathcal W(p_{AB})$, we will have it for all $p_{AB|U}$ by using a continuity argument and approaching $p_{AB|U}$ with elements of $\mathcal W(p_{AB})$.

Take some $p_{U|AB}$ with $p_{AB|U}\in \mathcal W(p_{AB})$. For every $U=u$, let
$$v_{AB|U=u}=p_{AB|U=u}-p_{AB}.$$ 
Further, let $v_{AB|U}$ be the random vector that is a function of $U$ and takes the vector $v_{AB|U=u}$ when $U=u$. 
we have
\begin{align*}
\E_U[\Upsilon(p_{AB|U})]\nonumber & = \Upsilon(p_{AB}) + \E_U\big[D^{(1)}_{v_{AB|U}}(p_{AB})\big] + \frac{1}{2}\E_U\big[D^{(2)}_{v_{AB|U}}(p_{AB})\big] + O(\E_U[\|v_{AB|U}\|_1^3])\\
& = \Upsilon(p_{AB}) + \frac{1}{2}\E_U\big[D^{(2)}_{v_{AB|U}}(p_{AB})\big] + O(\E_U[\|v_{AB|U}\|_1^3]),
\end{align*}
where in the second line we use $\E_U[v_{AB|U}] = \E_U[p_{AB|U}]- p_{AB}=0$, and that $D^{(1)}_{v}(p_{AB})$ is linear in $v$.
Thus we have
\begin{align}
I(U; AB) -\lambda_1I(U; A) - \lambda_2 I(U; B) & = \E_U[\Upsilon(p_{AB|U})] - \Upsilon(p_{AB})\nonumber\\
&= \frac{1}{2}\E_U\big[D^{(2)}_{v_{AB|U}}(p_{AB})\big] + O(\E_U[\|v_{AB|U}\|_1^3]).\label{eq:mutual-inf-hessian}
\end{align}

Now suppose that $(\lambda_1, \lambda_2)\in \fS(A, B)$. This implies that $D^{(2)}_v(p_{AB})\geq 0$. Then using~\eqref{eq:mutual-inf-hessian}, for every $p_{U|AB}$ we have
\begin{align*}I(U; AB) -\lambda_1I(U; A) - \lambda_2 I(U; B) + O(\E_U[\|p_{AB|U}-p_{AB}\|_1^3])\geq 0.\end{align*}

\vspace{.15in}
\noindent
\emph{Proof of $\fS_{\text{\rm i}}(A, B)\subseteq \fS(A, B)$}:
Now suppose that $(\lambda_1, \lambda_2)\in \fS_{\text{\rm i}}(A, B)$. 
We may write~\eqref{eq:mutual-inf-hessian} for the particular distribution $p_{ABU_\epsilon}$ defined in the theorem.
For this distribution we have 
$$v_{AB|U_\epsilon=u}=p_{AB|U_\epsilon=u}-p_{AB} = \epsilon u p_{AB} \cdot f_{AB}.$$ 
Therefore,  
$$\frac{\epsilon^2}{2}D^{(2)}_{p_{AB}\cdot f_{AB}}(p_{AB}) + O(\epsilon^3\|f\|_1^3)\geq 0.$$
Since this inequality should hold in a neighborhood of $\epsilon=0$, we must have $D^{(2)}_{p_{AB}\cdot f_{AB}}(p_{AB})\geq 0$. As mentioned in Section~\ref{subsec:MC-HC-ribbons} we have
\begin{align}
D^{(2)}_{p_{AB}\cdot f_{AB}}(p_{AB}) = \E[f^2] -\lambda_1 \E_A[(\E_{B|A}[f])^2)] -\lambda_2 \E_B[(\E_{A|B}[f])^2\geq 0.\label{eq:hessian-1}
\end{align}
Thus using Lemma~\ref{lem:def-MC-ribbon-zero} we conclude that $(\lambda_1, \lambda_2)\in \fS(A, B)$. Therefore, $\fS_{\text{\rm i}}(A, B)\subseteq \fS(A, B)$.

%***********************Proof rho-ribbon*****************************************
\section{Proof of Lemma~\ref{lem:ribbon-rho-CHSH}}\label{app:rho-chsh}

Any no-signaling box with binary inputs and outputs is determined by eight parameters. Indeed, we may write
\begin{align*}
q(a, b| x,y)=\frac{1}{4}(1+(-1)^a\alpha_x+(-1)^b\beta_y + (-1)^{a+ b}\zeta_{xy})
\end{align*}
Then 
$$q(a|x,y)=\frac{1}{2}(1+(-1)^a\alpha_x), \qquad q(b|x,y)=\frac{1}{2}(1+(-1)^b\beta_y),$$
and $q(a, b|x,y)$ is no-signaling. The fact that $q(a, b| x,y)$'s are non-negative is equivalent to
$$1-|\alpha_x - \beta_y| \geq \zeta_{xy} \geq |\alpha_x+ \beta_y|-1,$$
for all $x, y$. In particular we have $|\alpha_x|, |\beta_y|\leq 1$.

The maximal correlation of a bipartite random variable with binary parts can be found in~\cite[Proof of Lemma 7]{Beigi12}. Indeed for every $x, y$ we have
\begin{align*}
\rho(q(a, b| x,y)) = \frac{|\zeta_{xy}-\alpha_x\beta_y|}{\sqrt{(1-\alpha_x^2)(1-\beta_y^2)}}\,,
\end{align*}
where we put $\frac{0}{0}=0$.
Then
\begin{align}\label{eq:rho-f23}
\rho(q)=\max_{x, y} \frac{|\zeta_{xy}-\alpha_x\beta_y|}{\sqrt{(1-\alpha_x^2)(1-\beta_y^2)}}.
\end{align}

On the other hand it is not hard to see that
\begin{align}\label{eq:chsh-q-zeta}
\CHSH(q)= \frac{1}{4}\sum_{x,y} \frac{1+(-1)^{xy}\zeta_{x,y}}{2}\geq \frac{1+\eta}{2}.
\end{align}
Next, we show that for every $\alpha_x, \beta_y, \zeta_{xy}$ with the above conditions, equation~\eqref{eq:chsh-q-zeta} implies that $\rho(q)$ given by~\eqref{eq:rho-f23} is at least $\eta$.

We need to show that if $\eta\geq 1/\sqrt 2$,
\begin{align}\label{eq:range-alpha-beta}
1-|\alpha_x - \beta_y| \geq \zeta_{xy} \geq |\alpha_x+ \beta_y|-1,\qquad \qquad \forall x, y,
\end{align}
and 
\begin{align}\label{eq:sum-zeta-eta}
\sum_{x, y} (-1)^{xy}\zeta_{xy}\geq 4\eta,
\end{align}
then 
\begin{align}\max_{x, y} \frac{|\zeta_{xy}-\alpha_x\beta_y|}{\sqrt{(1-\alpha_x^2)(1-\beta_y^2)}}\geq \eta.\label{eq:maxrho-1}\end{align}

Observe that if $|\alpha_x|=1$, for some $x$, then from~\eqref{eq:range-alpha-beta} we have $\zeta_{xy}=\alpha_x\beta_y$ for all $y$. This holds because $|\alpha_x|=1$ implies that  $1-|\alpha_x - \beta_y|=|\alpha_x+ \beta_y|-1=\alpha_x\beta_y$ in this case.  Similarly if $|\beta_y|=1$ for some $y$, then $\zeta_{xy}=\alpha_x\beta_y$ for all $x$. As a result, if for all pairs $(x, y)$ we have either $|\alpha_x|=1$ or $|\beta_y|=1$, then $\zeta_{xy}=\alpha_x\beta_y$ for all $x,y$. In this case
 by~\eqref{eq:sum-zeta-eta} we have
$$2{\sqrt 2}\leq 4\eta\leq \sum_{x, y}(-1)^{xy}\alpha_x \beta_y,$$
which is a contradiction since by Bell's inequality we know that left hand side is at most $2$ (note that $|\alpha_x|, |\beta_y|\leq 1$).
Thus in the following we assume that for at least one pair of $(x, y)$ we have $|\alpha_x|\neq 1\neq |\beta_y|$.

To get a contradiction suppose that 
$$\frac{|\zeta_{xy}-\alpha_x\beta_y|}{\sqrt{(1-\alpha_x^2)(1-\beta_y^2)}}< \eta, \qquad \quad \forall x,y.$$
Then $(-1)^{xy}(\zeta_{xy}-\alpha_x\beta_y) \leq \eta\sqrt{(1-\alpha_x^2)(1-\beta_y^2)}$ for all $x, y$; further this inequality is strict for the pairs $(x, y)$ with $|\alpha_x|\neq 1\neq |\beta_y|$. Therefore by the above discussion we have
$$\sum_{x, y} (-1)^{xy}\zeta_{xy}< \sum_{x, y} \bigg[(-1)^{xy}\alpha_x\beta_y + \eta\sqrt{(1-\alpha_x^2)(1-\beta_y^2)}\bigg].$$
Comparing with~\eqref{eq:sum-zeta-eta}, we conclude that
\begin{align}\label{eq:4eta-inner1}
4\eta< \sum_{x, y} \bigg[(-1)^{xy}\alpha_x\beta_y + \eta\sqrt{(1-\alpha_x^2)(1-\beta_y^2)}\bigg].
\end{align}
Let us define 
$$v_{x} = \begin{bmatrix}
\alpha_x\\\sqrt{1-\alpha_x^2}
\end{bmatrix},  w_{y} = \begin{bmatrix}
\beta_y\\ \sqrt{1-\beta_y^2}
\end{bmatrix},   M_{xy}=\begin{bmatrix}
(-1)^{xy} & 0\\ 0 & \eta
\end{bmatrix}.$$
Also define
$$\tilde v=\begin{bmatrix}
v_0\\ v_1
\end{bmatrix}, \qquad  \tilde w=\begin{bmatrix}
w_0\\ w_1
\end{bmatrix},\qquad   \widetilde M=\begin{bmatrix}
M_{00}& M_{01}\\
M_{10} & M_{11}
\end{bmatrix}.$$
Then \eqref{eq:4eta-inner1} is equivalent to 
$$4\eta < \tilde v^t \widetilde M\tilde w.$$
Using the fact that $\|\tilde v\| = \|\tilde w\|=\sqrt 2$, this means that $\|\widetilde M\|>2\eta$. We however have
$$\|\widetilde M\|=\max\{\sqrt 2, 2\eta\}.$$
Then we should have $\sqrt 2>2\eta$ which is in contradiction with $\eta\geq 1/\sqrt 2$. We are done.

%*****************************************************************************

%*****************************************************************************

\end{document}